\def\CALQAEXTENDED{1}
\documentclass[conference,twocolumn]{IEEEtran}

\usepackage{cite}
\usepackage{amsmath,amssymb}
\usepackage{graphicx}
\usepackage{booktabs}
\usepackage{tabularx}
\newcolumntype{L}{>{\raggedright\arraybackslash}X}
\newcolumntype{C}{>{\centering\arraybackslash}X}
\newcolumntype{R}{>{\raggedleft\arraybackslash}X}
\usepackage{url}
\usepackage[colorlinks=true,urlcolor=blue,linkcolor=black,citecolor=black,filecolor=black]{hyperref}
\usepackage{xspace}

\newtheorem{definition}{Definition}
\newtheorem{problem}{Problem}
\newtheorem{lemma}{Lemma}

\newtheorem{theorem}{Theorem}
\newtheorem{corollary}{Corollary}

\newcommand{\extURL}{\url{https://arxiv.org/abs/2610.13249}}
\newcommand{\codeURL}{\url{https://github.com/xiufengliu/calqa}}
\newif\ifextended
\ifdefined\CALQAEXTENDED\extendedtrue\else\extendedfalse\fi
\newcommand{\fullproof}[2]{\emph{Full proof:} \ifextended Section~#2.\else extended
version~\cite[Sec.~#1]{CALQAext}.\fi}

\newcommand{\E}{\mathbb{E}}
\renewcommand{\P}{\mathbb{P}}
\newcommand{\R}{\mathbb{R}}
\newcommand{\Z}{\mathbb{Z}}
\newcommand{\one}{\mathbf{1}}
\newcommand{\calqa}{CALQA\xspace}
\newcommand{\qa}{q^{a}}
\newcommand{\qb}{q^{b}}
\newcommand{\shat}{\hat{s}}
\newcommand{\qhat}{\hat{q}}
\newcommand{\dtv}{d_{\mathrm{TV}}}
\newcommand{\ip}[2]{\langle #1,#2\rangle}
\newcommand{\norm}[1]{\lVert #1\rVert_1}
\newcommand{\Sc}{\mathcal{A}}
\newcommand{\need}{\nu}

\begin{document}

\title{CALQA: Conformal Screening of Forecast-Ensemble Configurations from One Mergeable State\ifextended\\(Extended Version)\fi}

\author{\IEEEauthorblockN{Xiufeng Liu\IEEEauthorrefmark{1},
Ruyu Liu\IEEEauthorrefmark{1},
Yusen Liu\IEEEauthorrefmark{2},
Huan Huo\IEEEauthorrefmark{2},
Yanyan Yang\IEEEauthorrefmark{3},
Per Sieverts Nielsen\IEEEauthorrefmark{1}}
\IEEEauthorblockA{\IEEEauthorrefmark{1}Dept. of Technology, Management and Economics, Technical University of Denmark, 2800 Kongens Lyngby, Denmark\\
\IEEEauthorrefmark{2}School of Computer Science, University of Technology Sydney, Sydney, NSW 2007, Australia\\
\IEEEauthorrefmark{3}School of Computing, University of Portsmouth, Portsmouth PO1 3HE, UK\\
\IEEEauthorrefmark{1}\{xiuli, ruyli, pernn\}@dtu.dk,
\IEEEauthorrefmark{2}\{yusen.liu, huan.huo\}@uts.edu.au,
\IEEEauthorrefmark{3}Linda.Yang@port.ac.uk}}

\maketitle

\begin{abstract}
AutoML systems fit ensemble weights over model libraries; a search can score, per category,
market or time pane, thousands of weightings of the quantile-averaged (Vincentized) ensemble for the narrowest
conformally calibrated interval; exact calibration scores a group's rows per candidate. We present
\calqa, a fixed-size integer aggregate per group and pane that stores histograms of normalised block scores at a net
of weight vectors. One GROUP BY builds it; it merges by addition, retracts deletions by subtraction and scores
any candidate without reading a row. We prove that its score never falls below the exact split-conformal shift and
exceeds it by at most the net distance plus a grid gap; that coverage carries over to each fixed weight vector
and, for i.i.d.\ blocks and up to a union-bound slack, to a searched weight vector deployed with the state's
shift; and how small any such state can be. The state is thus a screen: it ranks all candidates, and exact
recomputation re-checks the top 100. On five public panels (604 groups, 10,000 candidates each) this two-stage
plan matches exact search on held-out blocks (99.95\% of its Winkler score), 1.3--59 times faster. On
two real panels of about $10^8$ rows (smart meters, retail sales) it is 44--54 times faster at 99.98--100.01\%
of that score and beats sample screens with the same re-check in regret at less total time. Picks deployed with their exact shift need a held-out
coverage check. The state does not pay off for small or greedy searches, off-net weights at ample memory or full
twelve-member support.
\end{abstract}


\section{Introduction}\label{sec:intro}

Retail, energy and cloud forecasting stacks store the probabilistic forecasts of several member models next to the
data they forecast~\cite{Fischer2012,Januschowski2024}, and they serve an ensemble. The usual combiner is quantile
averaging, or Vincentization (VI): the $p$-quantile of the ensemble is a weighted average of the members'
$p$-quantiles~\cite{Lichtendahl2013,Busetti2017}. AutoML systems already choose ensemble weights after training, over
libraries of fitted models, by greedy ensemble selection~\cite{Caruana2004,Feurer2015,Erickson2020}, also for
quantile forecasters~\cite{Shchur2023}. The weights are thus chosen after ingest: at every retraining such a job, or an
analyst working in the warehouse, can search for each group (a product category, a market, a quarter) over
thousands of candidate weight vectors $w$ for the one whose interval is narrowest while still covering, say, 90\% of
future blocks, where a block is all horizons of one forecast origin or the 24 hours of a delivery day. The target
needs calibration even when the averaged quantiles look calibrated row by row~\cite{Cooke2023,GrushkaCockayneJose2020}:
on our electricity-price panels the uncalibrated $(0.1,0.9)$ interval covers 0.813--0.823 of the rows but only
0.296--0.334 of the daily blocks.

Split conformal prediction gives the calibration for one fixed $w$: the interval is widened by an order statistic of
block scores on held-out calibration blocks~\cite{Vovk2012}. Computing it exactly reads every calibration row of the
group, so a search over $M$ candidates in $G$ groups costs $GM$ passes over the rows. The data-system options are in
Table~\ref{tab:options}. Stored per-configuration counts are a distributive aggregate that rolls up like a data-cube
measure~\cite{GuptaMumick1995,Gray1997Cube}, but only for configurations declared in advance. Mergeable conformal
summaries~\cite{LuFCP2023,Humbert2023,WenSimeoneXing2026} summarise the scores of one predictor. A block sample serves
any $w$, but it is valid only in coverage and cannot be refilled after a deletion without a rescan.

\begin{table}[!t]
\centering
\scriptsize
\caption{Options for calibrating post-hoc configurations per group and window.}
\label{tab:options}
\setlength{\tabcolsep}{2.2pt}
\begin{tabularx}{\columnwidth}{@{}llllL@{}}
\toprule
Option & Size & Post-hoc & Merge/ & Answer per $w$ \\
 & per group & $w$ & retract & \\
\midrule
Recompute & raw rows & yes & rows & exact \\
Stored counts (e) & $J{+}1$ per $w$ & declared & yes/yes & exact \\
One-score summary & one score & no & yes/-- & one predictor \\
Block sample & $s$ blocks & yes & yes/drop & valid in coverage \\
\calqa & $|N|(J{+}1){+}1$ & yes & yes/yes & dominates exact \\
\bottomrule
\end{tabularx}
\end{table}

\begin{figure}[!t]
\centering
\includegraphics[width=\columnwidth]{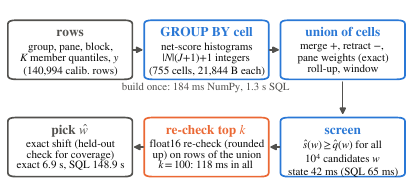}
\caption{Data flow of the two-stage configuration search, with measured times for the largest group of the configuration search in Section~\ref{sec:exp-search} (M4 Daily,
$K{=}4$, all 755 cells, $10^4$ Dirichlet candidates). NumPy search times come from the run of Table~\ref{tab:v2search};
the build times and those marked SQL (DuckDB) come from a separate run on the same data.
The state is built once; screening and re-check run per search.}
\label{fig:dataflow}
\end{figure}

\calqa is an aggregate that answers the calibration question for every candidate without reading a row. Once each
row is normalised by a scale that does not depend on $w$ and dominates the spread of the member quantiles, the
conformal score of a VI interval is $\frac12$-Lipschitz in $w$ in $\ell_1$, whatever the outcome and the level of the
series (Lemma~\ref{lem:lip}). \calqa therefore keeps, per group and time pane, histograms of block scores at a finite
net of weight vectors, and scores a candidate $w$ by the smallest net threshold plus half the $\ell_1$ distance to its
net point (Figure~\ref{fig:intuition}). The state is a vector of integers built by one GROUP BY. It merges by addition,
retracts deleted entities by subtraction and takes pane weights, so roll-ups, complements and sliding windows are
exact (Figure~\ref{fig:dataflow}).

The score is a safe screen. For every $w$ it lies between the exact split-conformal shift and that shift plus the net
distance and a grid gap (Theorem~\ref{thm:state}), so the state can rank all candidates, and exact recomputation from
the rows of the group needs to re-check only the best few. Over 604 groups of five public panels and 10,000
candidates per group, re-checking the top 100 picks configurations that do as well as exact search on held-out
blocks (99.95\% of the exact-search pick's block-Winkler score; 100.41\% for the state alone), at 1.3--59 times less time
than vectorised exact recomputation (Table~\ref{tab:v2search}). A block-sample screen with the same re-check is
slower and worse at the state's bytes, and larger samples win only at several times the time: on these panels the
plan is Pareto-optimal among the plans we ran, not dominant, and its gain is time, not memory. At scale the picture
sharpens. On 105,540,816 rows of the CER smart-meter panel~\cite{CER2012} and 88,786,880 rows of M5 retail
sales~\cite{Makridakis2022M5} (4,005 groups, 16-thread cluster nodes) the plan is 44--54 times faster than exact
search at 99.98--100.01\% of its held-out score, and it beats sample screens at one and four times the state's bytes
with the same re-check in regret at less total time (Table~\ref{tab:v2scale}).

The guarantees are proved, not measured. The interval returned for each fixed $w$, on or off the net, keeps the
weighted conformal coverage bound of Barber et al.~\cite{BarberCandes2023} with a gap that does not depend on $w$ or
on series levels; a weight vector picked by a search on the calibration blocks is covered, for i.i.d.\ blocks and
with the state's shift, up to a union-bound slack of order $\sqrt{\log|N|/n}$ over the net
(Corollary~\ref{cor:sel}). Picks deployed with their exact shift are not covered, and they cover only 0.84--0.90 of
held-out blocks per family. We also show that merge closure forces any deterministic dominating state to have order
$\epsilon^{-(1+d/2)}$ coordinates on a $d$-dimensional face of the weight simplex when the shift window is longer than
$\frac12$ (order $\epsilon^{-d/2}$ otherwise), and that net histograms need and achieve $\epsilon^{-(1+d)}$
(Table~\ref{tab:deploy} lists where the state does not pay off).

Our contributions are as follows.
\begin{itemize}
\item \emph{Operator.} Calibration screening over post-hoc ensemble weights, groups and windows as an in-database
aggregate: the net-score state, a group homomorphism under merge and retraction, built and queried in SQL
(Theorem~\ref{thm:state}, Figure~\ref{fig:sql}), and a two-stage search plan on top of it.
\item \emph{Guarantees and limits.} A deterministic sandwich against the exact shift for every $w$; transfers of
weighted conformal coverage to each post-hoc $w$ and, for i.i.d.\ blocks, to data-selected $w$
(Theorem~\ref{thm:cov}, Corollaries~\ref{cor:exch}--\ref{cor:sel}); a merge-closure lower bound for deterministic
dominating states and matching exponents for the net-histogram class (Lemma~\ref{lem:merge},
Theorems~\ref{thm:lb}--\ref{thm:tq}).
\item \emph{Evaluation.} An end-to-end configuration search, scored in and out of sample, against exact
recomputation, Lipschitz-pruned exact search and block-sample screens, in NumPy and DuckDB on five public panels
and in NumPy on two real panels of about $10^8$ rows each, with build, merge and query in both engines on
$10^8$ smart-meter rows; an equal-bytes regime map with random, realistic, declared and
neighbourhood weights; build, merge and query cost for many groups; forward-in-time decay; twelve-member families and
deletion (Section~\ref{sec:exp}).
\end{itemize}

\section{Related Work}

\emph{Combining and calibrating forecasts.} The linear pool is over-dispersed when its members are
calibrated~\cite{GneitingRanjan2013}; quantile averaging keeps the members' location-scale family and is
sharper~\cite{Lichtendahl2013,Busetti2017}, yet averaged quantiles can be overconfident~\cite{Cooke2023} and averaged
M4 interval endpoints were mostly too wide~\cite{GrushkaCockayneJose2020}. Fitted corrections~\cite{TaylorBunn1999}
restore calibration~\cite{Gneiting2007} but are refitted per configuration, and the Bates--Granger
formula~\cite{BatesGranger1969} predicts the error variance, not the coverage, of every weight vector.

\emph{Conformal prediction.} Split conformal prediction and its PAC version~\cite{Vovk2012} calibrate one fixed
predictor; weighted variants correct for covariate shift~\cite{Tibshirani2019}, and conformalized quantile regression
widens a quantile interval by one shift~\cite{Romano2019CQR}, as $I_o(w,s)$ does in units of $\sigma_o$. Barber et
al.\ bound the coverage loss without exchangeability by total-variation terms~\cite{BarberCandes2023};
Theorem~\ref{thm:cov} transfers their bound to each post-hoc weight vector. Two-layer data~\cite{LeeBarberWillett2023,
DunnWassermanRamdas2022} give our block-level exchangeability; time-series variants treat one score per
series~\cite{Stankeviciute2021,Chernozhukov2018,Diquigiovanni2021} or mixing data~\cite{Oliveira2024}, and adaptive
conformal inference keeps long-run coverage for any sequence~\cite{GibbsCandes2021,Zaffran2022}. EnbPI calibrates a
bootstrap ensemble~\cite{Xu2021EnbPI}, online aggregation combines the conformal sets of several
models~\cite{GasparinRamdas2024}, and online conformal prediction for panels treats non-exchangeable
series~\cite{TuGiesecke2026}. Selecting among finitely many conformal procedures costs a union-bound
slack~\cite{YangKuchibhotla2025}, or none when the selection is symmetrised~\cite{LiangZhuBarber2024}; aggregating
conformity scores over a weight grid~\cite{LuoZhou2024} or a function class~\cite{Bai2022} pays a uniform slack.
Corollary~\ref{cor:sel} applies the union bound to the net and carries it to all of $\Delta_K$.

\emph{Data systems and lower bounds.} What-if systems provision hypothetical updates~\cite{Deutch2013Caravan} or use
causal models~\cite{Galhotra2022HypeR}; our hypotheses are post-hoc weights of one aggregate. Approximate query
processing answers from samples~\cite{Agarwal2013BlinkDB,Chaudhuri2017AQP}, which our block-sample baselines are.
Mergeable summaries and quantile sketches summarise one value stream~\cite{Agarwal2012,Karnin2016,Masson2019,Gan2018},
and federated conformal prediction merges summaries of one model's scores~\cite{LuFCP2023,Humbert2023,Kang2024RobFCP},
also with T-Digests for prespecified groups~\cite{WenSimeoneXing2026}; none is indexed by a continuous weight vector.
Forward decay makes decayed aggregates mergeable~\cite{CormodeForwardDecay2009,CormodeTirthapuraXu2007,
CormodeTirthapuraXu2009}. Data cubes and materialised views keep distributive aggregates current under inserts and
deletes~\cite{Gray1997Cube,Harinarayan1996,GuptaMumick1995,Koch2014DBToaster}, non-distributive ones need auxiliary
state~\cite{Palpanas2002}, and stored per-configuration counts are such a view (baseline (e)). Quantile summaries
have lower bounds in the comparison model~\cite{CormodeVesely2020} and beyond~\cite{GuptaSinghalWu2024}, turnstile
algorithms may be taken to be linear~\cite{LiNguyenWoodruff2014}, and relative approximations bound range
counting~\cite{HarPeledSharir2011}. Our lower bound combines an $\epsilon$-capacity
count~\cite{KolmogorovTikhomirov1959} with the fuzzy-boundary counting of approximate range
searching~\cite{AryaMount2000}.

\section{Preliminaries and Problem}\label{sec:prelim}

\emph{Running example.} The German day-ahead panel (DE, Table~\ref{tab:v2datasets}) holds, for each delivery
day, 24 hourly rows with the $(0.1,0.9)$ quantile forecasts of $K=4$ members and the realised price. A block is one
delivery day and a group is one quarter, with 34--54 calibration days. On the forward test days of Section~\ref{sec:exp-time} the
uncalibrated VI interval covers 0.818 of the hours but only 0.296 of the days (Table~\ref{tab:v2temporalfwd}), so a
guarantee for whole days needs a calibrated shift. A configuration search then tries, for every quarter, $10^4$
weight vectors and keeps the one with the narrowest calibrated interval; computed exactly, each of the $10^4$ shifts
scores all 24 rows of every calibration day of the quarter. We return to this example in Sections~\ref{sec:method}
and~\ref{sec:exp}.

\emph{Rows and the aggregate.} A row $o=(\qa,\qb,y)$ holds the quantile forecasts $\qa,\qb\in\R^K$ of $K$
members at a central pair of levels $(a,1-a)$ and the outcome $y\in\R$. For a weight vector $w$ in the
simplex $\Delta_K=\{w\in\R^K: w\ge 0,\ \sum_k w_k=1\}$ the VI ensemble has the interval
$[\ell_o(w),h_o(w)]$ with $\ell_o(w)=\ip{w}{\qa}$ and $h_o(w)=\ip{w}{\qb}$. Each row has the scale
$\sigma_o=\max(\rho_o,\mu_o,\sigma_0)$, where $\rho_o$ is the larger of the ranges $\max_k\qa_k-\min_k\qa_k$
and $\max_k\qb_k-\min_k\qb_k$, $\mu_o=\frac1K\sum_k(\qb_k-\qa_k)$, and $\sigma_0>0$ is a constant;
we call this scale rule assumption B3. $\sigma_o$ does not depend on $w$. The row score and the calibrated
interval with shift $s\in\R$ are
\begin{align}
R_o(w)&=\max\bigl(\ell_o(w)-y,\ y-h_o(w)\bigr)/\sigma_o,\label{eq:score}\\
I_o(w,s)&=[\ell_o(w)-s\sigma_o,\ h_o(w)+s\sigma_o],\notag
\end{align}
with $I_o(w,+\infty)=\R$ and $I_o(w,-\infty)=\emptyset$; $I_o$ is the conformalized quantile regression
interval~\cite{Romano2019CQR} with the shift in units of $\sigma_o$. Since $\sum_k w_k=1$,
$R_o(w)=\max(\ip{w}{a_o},\ip{w}{b_o})$ with the normalised residuals $a_o=(\qa-y\one)/\sigma_o$ and
$b_o=(y\one-\qb)/\sigma_o$.

\emph{Blocks.} The unit of exchangeability is a block $U$, a nonempty finite sequence of rows, for example all
horizons of one forecast origin of one series; the block space is the disjoint union
$\bigcup_{\ell\ge1}(\R^{2K+1})^\ell$ with its Borel $\sigma$-algebra. A block rule $j$ picks
$j(\ell)\in\{1,\dots,\ell\}$, and the block score $S_U(w)$ is the $j(|U|)$-th smallest of $R_o(w)$, $o\in U$. With
$j(\ell)=\ell$, the rule of all our experiments, $S_U(w)\le s$ means that every row of $U$ is covered. For reals
$x_1,\dots,x_\ell$, $x_{(j)}$ denotes the $j$-th smallest, counted with multiplicity.

\emph{Exact shift.} Fix a miscoverage level $\alpha\in(0,1)$. Calibration blocks $U_1,\dots,U_n$ carry weights
$\omega_i\ge0$ and the test block a weight $\omega_{n+1}>0$. Let $\Omega=\sum_{i\le n}\omega_i$,
$\need=(1-\alpha)(\Omega+\omega_{n+1})$ and $C_w(s)=\sum_{i\le n}\omega_i\one\{S_{U_i}(w)\le s\}$. The exact shift is
\begin{equation}\label{eq:exact}
\qhat(w)=\min\{s\in\R: C_w(s)\ge\need\}\quad(+\infty\ \text{if}\ \need>\Omega).
\end{equation}
With unit weights, $\need=(1-\alpha)(n+1)$ and $\qhat(w)=s^*(w)$, the $k$-th smallest block score with
$k=k(n)=\lceil(1-\alpha)(n+1)\rceil$: the split-conformal shift of the configuration $w$~\cite{Vovk2012}. With
$\omega_i=\lambda^{\mathrm{age}_i}$ and $\omega_{n+1}=1$ it is the weighted threshold of Tibshirani et
al.~\cite{Tibshirani2019} and Barber et al.~\cite{BarberCandes2023}. For singleton blocks we write $s^*_A(w)$ for the unit-weight shift of a row
multiset $A$. Table~\ref{tab:notation} collects the notation.

\begin{table}[!t]
\centering
\scriptsize
\caption{Notation.}
\label{tab:notation}
\setlength{\tabcolsep}{2pt}
\begin{tabularx}{\columnwidth}{@{}lL@{}}
\toprule
$K$, $w$, $d$; $R_o$, $S_U$, $j$ & members, weights, face dim.; row/block score, rule\\
$\alpha$, $k(n)$, $\need$, $\Omega$ & miscoverage, split rank, weighted rank, total weight\\
$\qhat$, $s^*$, $\shat$ & exact (weighted) shift, unit-weight shift, state answer\\
$N$, $N_g$, $d_m(w)$ & net, lattice net, $\ell_1$ distance of $w$ to net point $w_m$\\
$G=\{s_1<\dots<s_J\}$, $\gamma_G$ & grid of $J$ points, grid gap\\
$\Phi$, $\varphi$, $\lambda$, $\tau$, $M$ & state, increment, decay, time stamp, \#queries\\
$I=[s_{\mathrm{lo}},s_{\mathrm{lo}}+B]$, $\epsilon$ & shift window, tightness\\
$\zeta$, $T_{\mathcal U}$, $p$ & gadget vector, number of types, face size bound\\
$\kappa$, $\delta$ & $\lfloor(B-\frac c2\Delta_{\mathcal U})/\delta\rfloor$, step $c/g^2$ (Thm.~\ref{thm:lb})\\
$\delta'$, $\epsilon'$, $\xi_0$, $\xi_{\delta'}$ & confidence, PAC level, DKW slacks (Cor.~\ref{cor:sel})\\
\bottomrule
\end{tabularx}
\end{table}

\begin{definition}[State classes]\label{def:scheme}
A \emph{state} is a map $\Phi$ from finite multisets of rows or blocks to a set $\mathcal S$ with a query
$\shat(\Phi(A),|A|,w)$ that reads only its arguments. It is \emph{mergeable} if some $\oplus$ satisfies
$\Phi(A\uplus A')=\Phi(A)\oplus\Phi(A')$ for all $A,A'$, and \emph{integer-additive} if $\mathcal S=\Z^D$ and
$\Phi(A)=\sum_{o\in A}\varphi(o)$. It is \emph{valid} on $W\subseteq\Delta_K$ if $\shat\ge s^*_A(w)$ for every
$A$ and $w\in W$ (pointwise dominance), and \emph{$\epsilon$-tight} on $(W,I)$, $I=[s_{\mathrm{lo}},s_{\mathrm{lo}}+B]$
with $B\ge0$, if $\shat\le s^*_A(w)+\epsilon$ whenever $w\in W$ and $s^*_A(w)\in I$.
\end{definition}

\begin{problem}[What-if calibration query]\label{prob:whatif}
\emph{Ingest:} a rows table with a group key (series, market) and a time stamp per block; keep, per group and time
pane (a \emph{cell}), a state of fixed size. \emph{Query, after ingest:} a union $\mathcal C$ of cells (a roll-up,
complement or sliding window, or what remains after deleting entities) and $M$ weight vectors $w\in\Delta_K$ that
were not known at ingest. \emph{Answer:} shifts $\shat(w)$ for the blocks of $\mathcal C$ with (a)
$\shat(w)\ge\qhat(w)$, (b) a small excess $\shat(w)-\qhat(w)$, and (c) a cost that grows with the number of cells in
$\mathcal C$ but not with the rows they hold.
\end{problem}

Requirement (a) makes every returned interval contain the exact conformal interval, so every guarantee of
the exact shift carries over; (b) is what the lower bounds of Section~\ref{sec:bounds} constrain.

\emph{The workload: configuration search.} For a union $\mathcal C$ and a candidate set $\mathcal W\subset\Delta_K$
of size $M$, a search returns the minimiser over $\mathcal W$ of the calibrated normalised width
\begin{equation}\label{eq:width}
f(w)=\ip{w}{\bar\beta}+2\qhat(w),
\end{equation}
the mean of $(h_o(w)-\ell_o(w))/\sigma_o+2\qhat(w)$ over the calibration rows of $\mathcal C$, where
$\bar\beta\in\R^K$ is the mean of $(\qb-\qa)/\sigma_o$. In the running example $\mathcal C$ is a quarter and
$\mathcal W$ holds $10^4$ Dirichlet draws; evaluating $f$ exactly scores every calibration row of $\mathcal C$ for
every candidate, even when the rows are read once for all of them. Section~\ref{sec:method} answers Problem~\ref{prob:whatif} with a state and uses it to screen
$\mathcal W$, so that rows are read for a few candidates only.

\section{The \calqa State}\label{sec:method}

\begin{figure}[!t]
\centering
\includegraphics[width=\columnwidth]{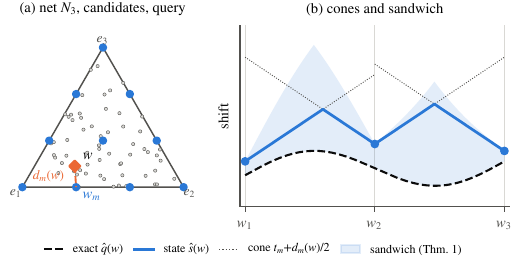}
\caption{Intuition (schematic, not data). (a) The weight simplex for $K{=}3$ with the lattice net $N_3$, candidates of a
search (grey) and a query $w$ at $\ell_1$ distance $d_m(w)$ from its nearest net point. (b) Along a segment through three net
points: the state stores thresholds $t_m$ (dots) at the net points; each spreads a cone of slope $\frac12$ in $\ell_1$
(Lemma~\ref{lem:lip}), and the state answers with their lower envelope $\hat s(w)$, which never falls below the exact
shift $\hat q(w)$ and exceeds it by at most $\min_m[d_m(w)+\gamma_G(r^*_m)]$ (shaded, Theorem~\ref{thm:state}(i)).}
\label{fig:intuition}
\end{figure}

\subsection{State and Query}

\begin{definition}[Net-score state]\label{def:state}
Fix a finite net $N=\{w_1,\dots,w_{|N|}\}\subset\Delta_K$ (default: the lattice
$N_g=\Delta_K\cap\frac1g\Z^K$) and a grid $G=\{s_1<\dots<s_J\}$ with bins $\beta_0=(-\infty,s_1]$,
$\beta_b=(s_b,s_{b+1}]$ and $\beta_J=(s_J,\infty)$. The increment $\varphi(U)\in\Z^{|N|(J+1)+1}$ of a block has
the entries $\varphi(U)[m,b]=\one\{S_U(w_m)\in\beta_b\}$ and a last entry $1$. For blocks with time stamps
$\tau_U<\tau$ and $\lambda\in(0,1]$, the state is $\Phi_\tau(A)=\sum_{U\in A}\lambda^{\tau-\tau_U}\varphi(U)$; for
$\lambda=1$ we write $\Phi(A)$. The query reads $\Omega$ as the last entry, sets
$C_m(s_j)=\sum_{b<j}\Phi_\tau[m,b]$ and $\need=(1-\alpha)(\Omega+1)$ (the value of $\need$ in~\eqref{eq:exact} for
$\omega_{n+1}=1$), and returns
\begin{equation}\label{eq:query}
t_m=\min\{s_j: C_m(s_j)\ge\need\},\quad \shat(w)=\min_{m}\bigl(t_m+\tfrac12 d_m(w)\bigr),
\end{equation}
with $d_m(w)=\norm{w-w_m}$ and $\min\emptyset=+\infty$. The grid gap is $\gamma_G(r)=\min\{s\in G:s\ge r\}-r$
for $r\le s_J$ and $+\infty$ otherwise.
\end{definition}

\emph{Running example, continued.} Take a DE quarter with $n=54$ calibration days. For a candidate $w$, exact split conformal computes the 24 row
scores of every calibration day, takes each day's maximum and returns the $k=\lceil0.9\cdot55\rceil=50$th smallest
of the 54 day scores; a search over $10^4$ candidates repeats this $10^4$ times over the $54\cdot24$ rows. The state
instead keeps, for each of the 84 points of $N_6$, a histogram of the 54 day scores over $J+1=65$ bins: $84\cdot65+1$
integers, 21,844 bytes, however many days the quarter holds. For each $w$ it reads the bin edge $t_m$ at which the
cumulative count reaches 50 at every net point and returns $\min_m(t_m+\frac12d_m(w))$. As every $w$ lies within
$\ell_1$ distance $\frac13$ of $N_6$ (Lemma~\ref{lem:round}), this answer never falls below the exact shift and
exceeds it by at most $\frac13$ plus a grid gap (Theorem~\ref{thm:state}(i)); merging two quarters adds their
histograms.

\begin{lemma}[Lipschitz transfer]\label{lem:lip}
Under B3, for every row $o$ with score~\eqref{eq:score}, block $U$ and $w,w'\in\Delta_K$:
(i) $|R_o(w)-R_o(w')|\le\frac12\norm{w-w'}$; (ii) $|S_U(w)-S_U(w')|\le\frac12\norm{w-w'}$;
(iii) for $s\in[-\infty,\infty]$, $S_U(w)\le s$ if and only if at least $j(|U|)$ rows $o$ of $U$ have
$y_o\in I_o(w,s)$.
\end{lemma}
\emph{Intuition.} As $\sum_kw_k=1$, the outcome $y$ cancels from $R_o(w)-R_o(w')$, which leaves weighted differences of member quantiles; each is at most half the quantile range $\rho_o\le\sigma_o$ times $\norm{w-w'}$. A $j$-th smallest value moves no more than the values it is taken from, and (iii) rewrites it as a count of covered rows.
\fullproof{A.1}{\ref{app:lip}}

\begin{lemma}[Weighted conservative quantisation]\label{lem:wc}
Let $v_1,\dots,v_n\in\R$, $\omega_i\ge0$, $\need>0$, $C(s)=\sum_i\omega_i\one\{v_i\le s\}$, $\Omega=\sum_i\omega_i$,
$r^*=\min\{s\in\R:C(s)\ge\need\}$ and $t=\min\{s\in G:C(s)\ge\need\}$ (each $+\infty$ if the set is empty).
(i) If $\need\le\Omega$, then $\{s:C(s)\ge\need\}=[r^*,\infty)$ and $r^*\in\{v_i\}$. (ii) $t=r^*+\gamma_G(r^*)$.
\end{lemma}
\emph{Intuition.} $C$ is a right-continuous step function that jumps only at the $v_i$, so its threshold is attained at some $v_i$; searching $G$ only rounds that threshold up to the next grid point.
\fullproof{A.2}{\ref{app:wc}}

\begin{theorem}[State: sandwich, algebra, size]\label{thm:state}
Let $A=\{U_1,\dots,U_n\}$ be any multiset of blocks with weights $\omega_i\ge0$ and $\omega_{n+1}>0$, and let
$N$, $G$ and $j$ be arbitrary. Let $t_m=\min\{s\in G:C_{w_m}(s)\ge\need\}$ with $C_w$ and
$\need=(1-\alpha)(\Omega+\omega_{n+1})$ as in~\eqref{eq:exact}, let $\shat(w)=\min_m(t_m+\frac12d_m(w))$, and let
$r^*_m=\qhat(w_m)$.
(i) For every $w\in\Delta_K$:
$\qhat(w)\le\shat(w)\le\qhat(w)+\min_m[d_m(w)+\gamma_G(r^*_m)]$.
(ii) $\Phi_\tau(A\uplus A')=\Phi_\tau(A)+\Phi_\tau(A')$, and $\Phi_\tau(A\setminus A')=\Phi_\tau(A)-\Phi_\tau(A')$ for
$A'\subseteq A$; $\Phi_{\tau+t'}(A)=\lambda^{t'}\Phi_\tau(A)$ for integers $t'\ge0$, so
$\Phi_{\tau+1}(A_{\le\tau})=\lambda(\Phi_\tau(A_{<\tau})+H_\tau)$, where $A_{<\tau}$ and $A_{\le\tau}$ hold the blocks
stamped before and up to $\tau$ and $H_\tau$ is the sum of the increments of the blocks stamped $\tau$.
(iii) With $\omega_i=\lambda^{\tau-\tau_i}$ and $\omega_{n+1}=1$, the query~\eqref{eq:query} computes these $t_m$ from
$\Phi_\tau$ alone, and $0<\omega_i\le\omega_{n+1}$.
(iv) For $\lambda=1$, $\Phi(A)\in\Z^{|N|(J+1)+1}$ with every entry in $[0,n]$; the number of entries depends
neither on $n$ nor on block sizes nor on the number of queried $w$, and $+$ and $-$ are exact.
(v) \emph{Pane state.} Let $A=P_1\uplus\dots\uplus P_R$ be split into panes with integer states $\Phi(P_r)$, and let
$\pi_r\in[0,1]$ be rational pane weights fixed in advance. Then $\Phi^\pi=\sum_r\pi_r\Phi(P_r)$ is the state of $A$
with block weights $\omega_i=\pi_{r(i)}$ ($U_i\in P_{r(i)}$), the query~\eqref{eq:query} on $\Phi^\pi$ returns the
$\shat$ of (i) for these weights and $\omega_{n+1}=1$, the counts $C_{w_m}(s_j)$, $\Omega$ and $\need$ it uses are exact
rationals, and adding, dropping or reweighting a pane changes $\Phi^\pi$ by one exact term.
\end{theorem}
\emph{Intuition.} (i) By Lemma~\ref{lem:lip}, a block within $t_m$ at $w_m$ is within $t_m+\frac12d_m(w)$ at $w$, so every cone of Figure~\ref{fig:intuition} lies above $\qhat(w)$; conversely, moving from $w$ to $w_m$, rounding up to the grid and coming back along the cone costs at most $d_m(w)+\gamma_G(r^*_m)$. (ii)--(v) hold because the state is a weighted sum of per-block increments whose cumulative bin counts are the weighted counts $C_{w_m}(s_j)$.
\fullproof{A.3}{\ref{app:state}}

\emph{Choosing the net and the grid.} Theorem~\ref{thm:state}(i) separates the two sources of excess. The net term
is at most the $\ell_1$ distance from $w$ to the net; for the lattice $N_g$ and a weight vector with $p$ nonzero
entries, Lemma~\ref{lem:round} bounds it by $r_p(g)\le p/(2g)$. Any finite net fixed in advance is allowed, so
weight vectors that will be queried often (equal weights of a subset, declared stacking weights) can be added as net
points, where the net term vanishes. A grid at pilot quantiles of the net scores puts its points where the
thresholds are. Neither choice may look at the calibration or test blocks (B0 below), since a data-dependent net or
grid would differ across shards and break both merging and the rank argument. For $\lambda<1$ the query is finite
only if $\need\le\Omega$, i.e.\ $\Omega\ge(1-\alpha)/\alpha$; with one block per age $\Omega<\lambda/(1-\lambda)$, so at
$\alpha=0.1$ decay needs $\lambda>0.9$.

\subsection{SQL Realisation}

\begin{figure}[!t]
\centering
\scriptsize
\begin{verbatim}
-- rows(grp,pane,blk,qa1..qaK,qb1..qbK,y,sig),
-- net(nid,w1..wK), grid(bin,s): bin 0..J-1 for
-- s_1<..<s_J, bin J = overflow; "+..+": K terms
CREATE TABLE st AS WITH bs AS (
  SELECT r.grp, r.pane, r.blk, n.nid,
    max(greatest(n.w1*(r.qa1-r.y)+..+n.wK*(r.qaK-r.y),
                 n.w1*(r.y-r.qb1)+..+n.wK*(r.y-r.qbK))
        / r.sig) AS sb          -- block rule: all rows
  FROM rows r CROSS JOIN net n GROUP BY ALL),
bb AS (SELECT bs.*, coalesce(g.bin, J) AS bin
  FROM bs ASOF LEFT JOIN grid g ON bs.sb <= g.s)
SELECT grp, pane, nid, bin, count(*) AS cnt
FROM bb GROUP BY ALL;
CREATE TABLE st_n AS SELECT grp, pane,  -- blocks
  count(DISTINCT blk) AS n FROM rows GROUP BY ALL;
-- query over a union u(uid,grp,pane)
WITH kk AS (SELECT uid, sum(n) AS n,  -- k=ceil(.9(n+1))
  (9*(sum(n)+1)+9)//10 AS k FROM st_n
  JOIN u USING (grp,pane) GROUP BY uid),
s AS (SELECT uid, nid, bin, sum(cnt) AS cnt FROM st
           JOIN u USING (grp,pane) GROUP BY ALL),
c AS (SELECT *, sum(cnt) OVER (PARTITION BY uid, nid
      ORDER BY bin) AS cum FROM s WHERE bin < J),
tb AS (SELECT c.uid, nid, min(bin) AS b FROM c
       JOIN kk USING (uid) WHERE cum >= k GROUP BY ALL),
tt AS (SELECT kk.uid, n.nid, CASE WHEN kk.k > kk.n
       OR tb.b IS NULL THEN 'inf'::DOUBLE ELSE g.s END AS tm
  FROM kk CROSS JOIN net n LEFT JOIN tb ON tb.uid=kk.uid
  AND tb.nid=n.nid LEFT JOIN grid g ON g.bin=tb.b)
SELECT tt.uid, q.qid, min(tt.tm + 0.5*(abs(q.w1-n.w1)
       +..+abs(q.wK-n.wK))) AS shift
FROM tt JOIN net n USING (nid) CROSS JOIN queries q
GROUP BY ALL;       -- one row per (uid,qid), +inf allowed
\end{verbatim}
\vspace{-6pt}
\caption{Build and query of \calqa (DuckDB), unit weights, block rule $j(\ell)=\ell$, $\alpha=0.1$; \texttt{st\_n}
counts blocks per cell and \texttt{kk} holds $n$ and $k$ per union. Every net point is kept by the LEFT JOINs, with $t_m=+\infty$ when
$k>n$ or the $k$-th score lies in the overflow bin, so a query whose net points are all $+\infty$ returns $+\infty$.}
\label{fig:sql}
\end{figure}

Figure~\ref{fig:sql} shows the statements, which are generated for each $K$; the executed files are in the code repository
(Section~\ref{sec:concl}). The rows table holds one row per series, time and block with the member
quantiles, the outcome, the group, pane and block keys and the precomputed scale $\sigma_o$ (wide schema; a long
(member, value) schema replaces the $K$-term sums by a join and a GROUP BY). A cross join with the net table computes
$R_o(w_m)$ for every row and net point; a MAX per block and net point gives $S_U(w_m)$ for $j(\ell)=\ell$ (a window
rank gives a general $j$); an as-of join with the grid assigns the bin, with scores above $s_J$ in the overflow bin
$J$, and a COUNT per (group, pane, net point, bin), with the block count, is the state. Roll-up is a SUM over cells, a
complement is the total minus the group, appending a batch adds its state, deleting entities subtracts the
increments of their blocks, and an expired pane is subtracted or dropped (Theorem~\ref{thm:state}(ii)). The query
sums the states of the $|\mathcal C|$ cells of the union in $O(|\mathcal C||N|J)$, computes the $t_m$ with a cumulative
count in $O(|N|J)$ and each $\shat(w)$ in $O(|N|K)$, so $M$ post-hoc configurations cost
$O(|\mathcal C||N|J+M|N|K)$ without touching a row, against $O(n_{\mathrm{rows}}|N|K)$ for the build. Decayed
and pane-weighted queries replace the count by the weighted sum and $k$ by
$\need=(1-\alpha)(\Omega+1)$ (Theorem~\ref{thm:state}(iii),(v)). The state takes $4(|N|(J+1)+1)$ bytes with 4-byte
counts: 21,844 bytes for $K=4$, $g=6$ ($|N|=84$) and $J=64$, and 32,764 bytes for $K=6$, $g=4$ ($|N|=126$); a sparse
encoding keeps only nonzero counts. Stored per-configuration counts, baseline (e), keep $J+1$ bins per declared
configuration, 260 bytes each, so the same bytes hold 84 (resp.\ 126) configurations, exact for those only.

\emph{Screening a configuration search.} The search of Section~\ref{sec:prelim} minimises $f$ of~\eqref{eq:width}
over $\mathcal W$. The $K$ sums and the row count behind $\bar\beta$ are $K+1$ additive real-valued entries kept beside
the integer state ($4(K+1)$ bytes), so $\ip{w}{\bar\beta}$ is computed from the state for every union. With
$e(w)=\min_m[d_m(w)+\gamma_G(r^*_m)]$, the state score $\hat f(w)=\ip{w}{\bar\beta}+2\shat(w)$ satisfies
$f(w)\le\hat f(w)\le f(w)+2e(w)$: Theorem~\ref{thm:state}(i) times two plus the same finite term (both sides are
$+\infty$ if $\qhat(w)=+\infty$). A \emph{state-only} plan returns the minimiser $\hat w_1$ of $\hat f$ with the shift
$\shat$. A \emph{two-stage} plan ranks $\mathcal W$ by $\hat f$, computes $f$ from the rows of the union for the $k$
best (exactly, or from score-sufficient rows rounded up to float16 as in Section~\ref{sec:exp-search}, an upper bound $f_{16}\ge f$) and returns
the best of those, $\hat w_2$, with its exact shift; the screen costs $O(|\mathcal C||N|J+M|N|K)$ and rows are scored
for $k$ candidates instead of $M$ (Figure~\ref{fig:dataflow}). A candidate never looks better to the screen than it
is. For the minimiser $w^\star$ of $f$ over $\mathcal W$,
$f(\hat w_1)\le\hat f(\hat w_1)\le\hat f(w^\star)\le f(w^\star)+2e(w^\star)$. With an exact re-check, $\hat w_1$ is
re-checked, so $f(\hat w_2)\le f(\hat w_1)$: the two-stage pick is never worse than the state-only pick. It attains
$f(w^\star)$ whenever at most $k$ candidates have $f(v)\le\hat f(w^\star)$, since every $v$ ranked at or ahead of
$w^\star$ has $f(v)\le\hat f(v)\le\hat f(w^\star)$; otherwise it may miss $w^\star$. With a float16 re-check,
$f(\hat w_2)\le f_{16}(\hat w_2)\le f_{16}(\hat w_1)$. Corollary~\ref{cor:sel} concerns the shift $\shat$, so a pick
returned with its exact shift is not covered by it.

\subsection{Coverage for Each Post-Hoc Weight Vector}

The guarantees use the following assumptions. \emph{B0:} the net, grid, block rule, level pair, $\sigma_0$,
$\alpha$, the weight function (a function of time stamps or of pane labels fixed in advance only), the block partition and the number and
time stamps of the calibration blocks are fixed before, and independently of, the block values (a disjoint
pilot may choose them). \emph{B1:} $0\le\omega_i\le\omega_{n+1}$. \emph{B2:} $w$ is fixed or independent of
$Z=(U_1,\dots,U_{n+1})$. B3 is the scale rule of Section~\ref{sec:prelim}. Let
$\mathrm{wt}_i=\omega_i/(\Omega+\omega_{n+1})$, let $Z^{(i)}$ be $Z$ with $U_i$ and $U_{n+1}$ exchanged, and for a map
$F$ applied to each block let $F(Z)^{(i)}$ swap the same two coordinates. The \emph{score-sufficient image} of a
block is $\Sc(U)=((a_o,b_o))_{o\in U}$, and $V^w=(S_{U_1}(w),\dots,S_{U_{n+1}}(w))$. The gaps are
\begin{align*}
G_w&=\textstyle\sum_{i\le n}\mathrm{wt}_i\,\dtv\bigl(V^w,(V^w)^{(i)}\bigr),\\
G_Z&=\textstyle\sum_{i\le n}\mathrm{wt}_i\,\dtv\bigl(Z,Z^{(i)}\bigr),\\
G_{\Sc}&=\textstyle\sum_{i\le n}\mathrm{wt}_i\,\dtv\bigl(\Sc(Z),\Sc(Z)^{(i)}\bigr).
\end{align*}

\begin{theorem}[Coverage for each post-hoc $w$; transfer of Barber et al.~\cite{BarberCandes2023}]\label{thm:cov}
Assume B0--B3 and let $\shat$ be computed as in Theorem~\ref{thm:state} from the weighted counts of
$U_1,\dots,U_n$ (from the state itself for decayed or pane weights, Theorem~\ref{thm:state}(iii),(v)). For any joint
law of $Z$ and every fixed $w\in\Delta_K$, on or off the net,
\begin{align*}
&\P\bigl(\text{at least } j(|U_{n+1}|) \text{ rows } o\in U_{n+1} \text{ have } y_o\in I_o(w,\shat(w))\bigr)\\
&\qquad\ge 1-\alpha-G_w\ \ge\ 1-\alpha-G_{\Sc}\ \ge\ 1-\alpha-G_Z .
\end{align*}
$G_{\Sc}$ does not depend on $w$, the net, the grid or the block rule, so $\sup_wG_w\le G_{\Sc}$. If each block is
replaced by $c_i+b_iU_i$ ($b_i>0$, possibly random, applied to $\qa,\qb,y$ of its rows) and $\sigma_0$ binds in
no row before or after, $G_{\Sc}$ is unchanged. For random $w$ independent of $Z$, the first inequality holds with
$\E[G_w]$ in place of $G_w$; the others are unchanged.
\end{theorem}
The swap argument and the weighted threshold are those of~\cite{BarberCandes2023}, and weighted
thresholds go back to Tibshirani et al.~\cite{Tibshirani2019}; what is new is the
domination step that carries them off the net and the $w$-free, scale-invariant gap $G_{\Sc}$.

\begin{corollary}[Exchangeable blocks]\label{cor:exch}
Under B0, B2 and B3, if $Z$ is exchangeable, the coverage is at least $1-\alpha$ for each $w\in\Delta_K$ and all
weights with B1; with $\lambda=1$ this is the split-conformal guarantee for each post-hoc $w$.
\end{corollary}

\emph{Intuition.} As $\shat(w)\ge\qhat(w)$ (Theorem~\ref{thm:state}(i)), a miss at $w$ is a miss of the exact weighted procedure on the block scores $V^w$, which the weighted swap argument bounds by $\alpha+G_w$. $V^w$ is a blockwise function of $\Sc(Z)$, and $\Sc(Z)$ of $Z$; total variation cannot grow under such maps, and $\Sc$ is unchanged by the affine maps when $\sigma_0$ does not bind. Under exchangeability every gap is zero (Corollary~\ref{cor:exch}).
\fullproof{A.4}{\ref{app:cov}}

Decayed and truncated (sliding-window) weights satisfy B1 (Theorem~\ref{thm:state}(iii),(v)); distance to the net
enters the width, never the coverage. For blocks of series with disjoint level ranges the laws of $Z$ and $Z^{(i)}$ are mutually
singular and that term of $G_Z$ equals one, while
$G_{\Sc}$ charges only for drift in the law of the normalised residuals. Neither gap can be estimated from one realised path, so Theorem~\ref{thm:cov} is a
robustness bound, not a certificate.

\emph{Weights selected from the data.} B2 excludes a $w$ tuned on the calibration blocks. Because the state
answers every $w$ through the $|N|$ net thresholds, the finite-selection argument of Yang and
Kuchibhotla~\cite{YangKuchibhotla2025} over the net covers such $w$ for i.i.d.\ blocks.

\begin{corollary}[Data-selected $w$, after~\cite{YangKuchibhotla2025}]\label{cor:sel}
Assume B0, B3, unit weights and $U_1,\dots,U_{n+1}$ i.i.d., and let $\hat w$ be any measurable function of
$U_{1:n}$ and of a random element $\vartheta$ independent of $Z$, into $\Delta_K$. Let $\xi_0=\sqrt{\log(2|N|)/(2n)}$
and, for $\delta'\in(0,1)$, $\xi_{\delta'}=\sqrt{\log(2|N|/\delta')/(2n)}$. (i) The coverage event of Theorem~\ref{thm:cov} at $\hat w$ has
probability at least $1-\alpha-\xi_0-\sqrt{\pi/(8n)}$. (ii) With probability at least $1-\delta'$ over $U_{1:n}$, the
conditional coverage given $U_{1:n}$ is at least $1-\alpha-\xi_{\delta'}$ for all $w\in\Delta_K$ simultaneously.
(iii) If the query uses a rank $k'\ge\max(k(n),n(1-\epsilon'+\xi_{\delta'}))$ in place of $k(n)$, then with probability
at least $1-\delta'$ the conditional coverage is at least $1-\epsilon'$ for all $w$ simultaneously, and
Corollary~\ref{cor:exch} still holds.
\end{corollary}
\emph{Intuition.} At any $w$, data-selected or not, $\shat(w)$ is a cone over one of $|N|$ net thresholds $t_m$, and a fresh block within $t_m$ at $w_m$ is covered at $w$ (Lemma~\ref{lem:lip}). Coverage is therefore at least $k/n$ minus the largest deviation of $|N|$ empirical distribution functions, which the Dvoretzky--Kiefer--Wolfowitz inequality with Massart's constant~\cite{Massart1990} and a union bound over the net control; a larger rank only raises every $t_m$.
\fullproof{A.5}{\ref{app:sel}}

Parts (ii)--(iii) hold simultaneously over $\Delta_K$ and part (i) at a data-selected $\hat w$, such as the pick of a
state-only search deployed with $\shat$; Theorem~\ref{thm:cov} and Corollary~\ref{cor:exch} hold for each fixed
$w$. The slack is the price of selection: at the calibration sizes of Section~\ref{sec:exp-bytes} and $|N|=84$, part (i) costs 0.016 (M4 Daily)
to 0.129 (DE), at the group sizes of Section~\ref{sec:exp-search} 0.022 for the largest group and 0.26--0.51 for the 478 groups with 20--72
calibration blocks, and the rank of (iii) at $(\epsilon',\delta')=(0.1,0.05)$ is finite only for $n\ge406$ (427 for
$|N|=126$).

\section{State Size: Optimal in Exponent Within the Net-Histogram Class}\label{sec:bounds}

This section bounds the size of any state that meets requirements (a) and (b) of Problem~\ref{prob:whatif}
in the worst case; the net-histogram class matches its own lower bound in exponent, while for general states a gap of
$d/2$ in the exponent stays open. Rows are \emph{legal} if $\qa\le\qb$ componentwise (L). A block state that is valid and tight
for all multisets of blocks is, on singleton blocks ($j(1)=1$, block score equal to row score), a row state with
the same properties, so lower bounds for rows apply to block states. The classes are those of
Definition~\ref{def:scheme}: \emph{deterministic} states that dominate $s^*$ on \emph{every} multiset. Summaries
that are valid only in coverage, such as block samples with split conformal, and randomized summaries lie outside
them.

\begin{lemma}[Gadgets and merge closure]\label{lem:merge}
Assume L and B3. (a) For $\zeta\in\R^K$ with $\min_k\zeta_k\ge-\frac12$ and $\max_k\zeta_k-\min_k\zeta_k\le1$,
the row $o_\zeta=(\sigma_0\zeta,\sigma_0(\zeta+\one),0)$ is legal and $R_{o_\zeta}(v)=\ip{v}{\zeta}$ for all
$v\in\Delta_K$; with $\zeta=r\one$, $R\equiv r$ for every $r\ge-\frac12$.
(b) Let $(\Phi,\shat)$ be mergeable, valid on $W$ and $\epsilon$-tight on $(W,I)$ with $s_{\mathrm{lo}}>-\frac12$. Let
$|A|=|A'|$ and $\Phi(A)=\Phi(A')$, let $w\in W$ and $s\le s_{\mathrm{lo}}+B$ be such that every $o\in A$ with
$R_o(w)\le s$ has $R_o(w)\ge s_{\mathrm{lo}}$, and let $\epsilon''>\epsilon$. Then
$n_s:=\#\{o\in A:R_o(w)\le s\}\le n'_s:=\#\{o\in A':R_o(w)\le s+\epsilon''\}$.
\end{lemma}
\emph{Intuition.} (a) gives rows whose score is linear in $v$. For (b), both multisets are padded with the same low and high rows so that the $k$-th smallest score lies in $[s_{\mathrm{lo}},s]$ for one and above $s+\epsilon''$ for the other; mergeability keeps their states equal, so their common answer cannot be within $\epsilon$ of the first shift (tightness) and above the second (validity).
\fullproof{A.6}{\ref{app:merge}}

Merge closure thus forces a valid and tight state to determine, at every query point and threshold in the
window, the count $\#\{o:R_o(w)\le s\}$ up to an $\epsilon$-shift of the threshold: fuzzy-boundary counting in
the sense of Arya and Mount~\cite{AryaMount2000}, with sublevel sets of the score as ranges.

\begin{theorem}[General lower bound]\label{thm:lb}
Let $(\Phi,\shat)$ be mergeable, valid on $W$ and $\epsilon$-tight on $(W,I)$ with $s_{\mathrm{lo}}\ge0$, under L
and B3. Let $g\ge1$, $c\in(0,1]$, $\delta=c/g^2>\epsilon$, let $\emptyset\ne\mathcal U\subseteq N_g\cap W$ with
$\Delta_{\mathcal U}=\max_{\mathcal U}|u|^2-\min_{\mathcal U}|u|^2$ (Euclidean norm) and $B\ge\frac c2\Delta_{\mathcal U}$, and
$T_{\mathcal U}=(\lfloor(B-\frac c2\Delta_{\mathcal U})/\delta\rfloor+1)|\mathcal U|$.
(a) $\Phi$ takes at least $\binom{n+T_{\mathcal U}-1}{T_{\mathcal U}-1}$ values on multisets of size $n$, for every
$n\ge1$, so a $b$-bit state needs $b\ge\log_2\binom{n+T_{\mathcal U}-1}{T_{\mathcal U}-1}$. (b) An integer-additive
$\Phi$ has $D\ge T_{\mathcal U}-1$. (c) Let $W$ contain a $d$-face $\Delta_F=\{w\in\Delta_K:\mathrm{supp}\,w\subseteq F\}$
($|F|=d+1$, $d\ge1$), and let $\epsilon<\frac14$ and $B\ge\frac12$. Then $c=1$, $g=\lceil\epsilon^{-1/2}\rceil-1$ and
$\mathcal U=N_g\cap\Delta_F$ are admissible, and $T_{\mathcal U}\ge(\lfloor(B-\frac12)g^2\rfloor+1)\binom{g+d}{d}$,
which is $\Omega_d((B-\frac12)\epsilon^{-(1+d/2)}+\epsilon^{-d/2})$.
\end{theorem}
\emph{Intuition.} Each row type is the tangent plane of $-\frac c2|w|^2$ at a lattice point $u$, lifted by $j\delta$; each test compares a score at a lattice point with a threshold. On the lattice, score minus threshold is a multiple of $\delta>\epsilon$, so Lemma~\ref{lem:merge} applies to every test, and the type--test incidence matrix is unimodular. The state must therefore separate all multisets of types (a), a linear integer state needs $T_{\mathcal U}-1$ coordinates (b), and counting lattice points on a face gives (c).
\fullproof{A.7}{\ref{app:lb}}

The argument combines an $\epsilon$-capacity count~\cite{KolmogorovTikhomirov1959} with the rank argument for linear
states; Lemma~\ref{lem:merge} and the lattice-exact robustness of the tangent family are what make it apply to
conformal shift states. For one $w$, a one-point net histogram with grid spacing $\epsilon$ is
$\epsilon$-tight with $\lceil B/\epsilon\rceil+2$ coordinates (Theorem~\ref{thm:state}(i)), and Theorem~\ref{thm:tq}(c)
shows that this grid size is necessary for the class.

\begin{lemma}[Support-preserving rounding]\label{lem:round}
Let $w\in\Delta_K$ have support of size $p\ge2$ and $g\ge1$. The largest-remainder rounding $u\in N_g$ of $w$ has
$\mathrm{supp}\,u\subseteq\mathrm{supp}\,w$ and $\norm{w-u}\le r_p(g):=2\lfloor p/2\rfloor\lceil p/2\rceil/(pg)\le p/(2g)$;
$r_p(g)$ is nondecreasing in $p$.
\end{lemma}
\emph{Intuition.} Round $gw$ down and give the $r\le p-1$ leftover units to the largest remainders~\cite{ContHeidari2015}; only coordinates in the support change, and the $\ell_1$ error, at most $2r(p-r)/(pg)$, is largest at $r=\lfloor p/2\rfloor$.
\fullproof{A.8}{\ref{app:round}}

\begin{theorem}[Net histograms: optimal in exponent within their class]\label{thm:tq}
Consider the state of Definition~\ref{def:state} with unit weights, any net $N$ and grid $G$, and
$I=[s_{\mathrm{lo}},s_{\mathrm{lo}}+B]$ with $s_{\mathrm{lo}}>-\frac12$.
(a) For every $w_0\in\Delta_K$ ($K\ge2$) and $s_0\in I$ with $s_0\ge\frac12$ there is a multiset $A$ with
$s^*_A(v)=s_0+\frac12\norm{v-w_0}$ for all $v$, on which $\shat(w_0)-s^*_A(w_0)\ge\mathrm{dist}(w_0,N):=\min_m d_m(w_0)$.
(b) For every $x\in I$ there is a multiset with $s^*\equiv x$ and $\shat(w)-s^*(w)=\gamma_G(x)+\frac12\mathrm{dist}(w,N)$.
(c) If the state is $\epsilon$-tight on $(W,I)$ and $s_{\mathrm{lo}}+B\ge\frac12$, then $\mathrm{dist}(w,N)\le\epsilon$ on $W$,
$\gamma_G\le\epsilon$ on $I$, and $G$ has at least $\lceil B/\epsilon\rceil$ points. If $W$ contains a family $\mathcal F$ of
distinct $d$-faces, then $|N|\ge\max\bigl(1,|\mathcal F|(1-2(d+1)\epsilon)^d\bigr)(2\epsilon)^{-d}$ (the second term when
$2(d+1)\epsilon<1$), and the state has at least $|N|(\lceil B/\epsilon\rceil+1)+1$ coordinates.
(d) Let $W$ be a union of faces $\Delta_{F'}$ with $|F'|\le p$ ($p\ge2$), $N=N_g\cap W$, $r=r_p(g)$, $h>0$, and
$G=\{s_{\mathrm{lo}}-\frac r2+ih:0\le i<J\}$ with $J=\lceil(B+r)/h\rceil+1$. Then for every multiset of blocks and
every $w\in W$ with $s^*(w)\in I$, $s^*(w)\le\shat(w)\le s^*(w)+r+h$, while Corollary~\ref{cor:exch} and
Theorem~\ref{thm:state}(ii) still hold. On one $d$-face, $g=\lceil(d+1)/\epsilon\rceil$ and $h=\frac\epsilon2$ give
$\epsilon$-tightness with $O_d(B\epsilon^{-(1+d)})$ coordinates for $\epsilon\le\min(1,B)$, so, when
$s_{\mathrm{lo}}+B\ge\frac12$, the class needs and achieves $\Theta_d(B\epsilon^{-(1+d)})$ coordinates, with constants that
depend on $d$.
\end{theorem}
\emph{Intuition.} (a) builds rows whose exact shift is a cone of slope $\frac12$ around $w_0$, so at $w_0$ the state exceeds the exact shift by at least the distance from $w_0$ to the net; in (b) a constant score pays the grid gap. Tightness thus forces an $\epsilon$-covering net, which a volume count on each face bounds below, and grid gaps of at most $\epsilon$ (c); the lattice net with the rounding of Lemma~\ref{lem:round} and a uniform grid achieve $r+h\le\epsilon$ (d).
\fullproof{A.9}{\ref{app:tq}}

\emph{Net sizes and constants.} On the union $W_p$ of all faces with at most $p$ members,
$|N_g\cap W_p|=\sum_{j\le p}\binom Kj\binom{g-1}{j-1}$: a lattice point with support exactly $F$, $|F|=j$, is a
composition of $g$ into $j$ positive parts. At $K=12$ and $g=4$ this gives 210 points for pairs and 870 for triples,
against $\binom{15}{11}=1{,}365$ on $\Delta_{12}$. At $B=4$ and $\epsilon=1/40$ on one $d$-face, the upper count
$\min_g|N_g\cap\Delta_F|(J+1)+1$ of (d), with $h=\epsilon-r_{d+1}(g)$, exceeds the lower count
$(2\epsilon)^{-d}(\lceil B/\epsilon\rceil+1)+1$ of (c) by factors of 8.1, 25, 107 and 1,073 for $d=1,2,3,5$. For general
states the exponents are $1+d/2$ (Theorem~\ref{thm:lb}) against $1+d$ (Theorem~\ref{thm:tq}); whether a mergeable state
that is not a net histogram closes this gap of $d/2$ is open.

\emph{Consequences at $K=12$.} Take $\epsilon=0.1$ on the observed shift window $[0.08,1.75]$ ($B=1.67$). Theorem~\ref{thm:lb}(b)
with $g=3$, $c$ slightly above $0.9$ (so that $\delta=c/9>\epsilon$) and $\mathcal U$ the 132 points of $N_3$ with support
size two ($\Delta_{\mathcal U}=0$, $\kappa=16$) gives $T_{\mathcal U}-1=2{,}243$ integer coordinates for pairs; adding the
220 points with support size three ($\Delta_{\mathcal U}=\frac29$, $\kappa=15$) gives 5,631 for triples, which also bounds
$\Delta_{12}$ since it contains them. Theorem~\ref{thm:tq}(c), with $\lceil B/\epsilon\rceil=17$ and the 66 pair faces
($|N|\ge198$), the 220 triple faces ($|N|\ge880$) or $\Delta_{12}$ itself ($d=11$), gives 3,565, 15,841 and 878,906,251
coordinates for the net-histogram class (bottom of Table~\ref{tab:v2k12}). Full support is therefore out of reach for this class,
while a deterministic dominating integer-additive state of 22,524 bytes (5,631 four-byte coordinates) on $\Delta_{12}$
is not excluded; restricted families of pairs or triples are what the class can serve.

\section{Experiments}\label{sec:exp}

The experiments ask whether the state pays off in an end-to-end configuration search and what it costs in decision
quality (Section~\ref{sec:exp-search}), also on two real panels of about $10^8$ rows (Section~\ref{sec:exp-scale}); when, at equal bytes, it is the narrowest reliable summary for random, realistic,
declared and neighbourhood weights (Section~\ref{sec:exp-bytes}); what build, merge and query cost for one panel and for
many groups (Section~\ref{sec:exp-cost}); how it behaves forward in time (Section~\ref{sec:exp-time}); and what it
achieves with twelve members and under deletion (Section~\ref{sec:exp-k12}). Table~\ref{tab:deploy}, at the end of
this section, condenses the findings into when to use the state and when not.

\subsection{Experimental Setup}

\begin{table}[!t]
\centering
\scriptsize
\caption{Panels and member subsets.}
\label{tab:v2datasets}
\setlength{\tabcolsep}{3pt}
\begin{tabularx}{\columnwidth}{@{}Lrlrrrr@{}}
\toprule
Panel & Series & Block & Rows/ & Blocks & $n_{\mathrm{cal}}$ & Test \\
 & & & block & & & steps \\
\midrule
M4 Hourly & 414 & series$\times$origin & 48 & 2,070 & 1,463 & 4 \\
M4 Daily & 4,227 & series$\times$origin & 14 & 21,135 & 18,622 & 4 \\
M4 Weekly & 359 & series$\times$origin & 13 & 1,795 & 1,216 & 4 \\
EPF & 120 & market$\times$day & 24 & 3,640 & 3,100 & 609 \\
DE & 24 & day & 24 & 728 & 300 & 609 \\
\bottomrule
\end{tabularx}
\\[2pt]
\parbox{\columnwidth}{\scriptsize Subsets $K{=}4$: ets, theta, arima, lgbm\_q / m3\_arx, m5\_ewma, m9\_lear, m10\_lgbm\_q; $K{=}6$ adds snaive, ridge / m1\_sn\_day, m11\_lgbm\_cp.}
\end{table}

\emph{Panels and members.} Table~\ref{tab:v2datasets} lists the panels: M4 Hourly, Daily and
Weekly~\cite{Makridakis2020M4}, the five day-ahead markets of the price forecasting benchmark (EPF,
$5\times24$ delivery-hour series)~\cite{Lago2021} and the German day-ahead panel (DE)~\cite{Marcjasz2023}. Every
panel has a pool of twelve members (statistical models, LightGBM quantile regression, ridge and climatology) with
19 quantile levels; the main experiments use the subsets $K=4$ and $K=6$ of Table~\ref{tab:v2datasets}, and Section~\ref{sec:exp-k12} uses
all twelve. The pools contain no neural or foundation models. The central pair is $(0.1,0.9)$, $\alpha=0.1$, and the
block rule covers all rows of a block. No member uses an outcome at or after its forecast origin: M4 members,
including the global LightGBM models, are refit at each of five rolling origins per series, and EPF and DE members use
only days before the delivery day. A block is all rows of one forecast origin of one series (M4) or of one delivery
day of one market (EPF, DE); $n_{\mathrm{cal}}$ counts calibration blocks per random split, and the test steps are the
origins or days of the forward-in-time runs (Section~\ref{sec:exp-time}).

\emph{\calqa settings and baselines.} Nets are $N_6$ for $K=4$ ($|N|=84$) and $N_4$ for $K=6$ ($|N|=126$) with $J=64$
grid points at quantiles of the net scores of a pilot (10\% of the blocks for M4, the first 28 days for EPF and DE),
which only places the grid and is excluded from calibration and test (B0). Baselines: (a) bottom-$s$ block samples by
hash with split conformal, stored as raw float32 rows, compressed integer rows, score-sufficient rows $(a_o,b_o)$
rounded up to float16 (valid, since the scores only grow) and float16 with zstd, as in approximate query
processing~\cite{Agarwal2013BlinkDB,Chaudhuri2017AQP}; (b) a net-score sample that keeps the $|N|$ block
scores of each sampled block; (c) training-conditional (PAC) versions of (a), (b) and exact with the binomial rank of
Vovk~\cite{Vovk2012}; (d) exact per-$w$ conformal from all calibration rows, in SQL, as a NumPy loop and vectorised
from float16 or float32 rows; (e) stored per-configuration counts. In equal-bytes comparisons, $(g,J)$ and the sample
encodings were selected per budget on other random partitions of the same non-pilot blocks, symmetrically for all
methods. Timings run on a 2-core virtual machine with DuckDB 1.5.6 (2 threads) and NumPy, as medians of five warm
repeats (interquartile range at most 0.70 of the median in the many-group runs). Budgets are in KB of 1,024 bytes.

\emph{Metrics and implementation check.} Coverage is block coverage (all rows of the test block inside the interval);
the width ratio is the mean normalised width relative to exact per-$w$ conformal with the same data and weights.
Random block splits make the blocks exchangeable (Corollary~\ref{cor:exch}), so their coverages check the code only: over 84 panel $\times$ $K$ $\times$ $n$ cells the net
state covers 0.908--0.974, exact per-$w$ conformal 0.887--0.910, and no shift falls below the exact one.

\subsection{End-to-End Configuration Search}\label{sec:exp-search}

\begin{table}[!t]
\centering
\scriptsize
\caption{End-to-end configuration search: 604 groups in 34 families, 5 splits.}
\label{tab:v2search}
\setlength{\tabcolsep}{2.6pt}
\begin{tabularx}{\columnwidth}{@{}Lrrrrrrr@{}}
\toprule
 & Test & Regret & Time & Speed- & & \multicolumn{2}{c}{Large groups} \\
\cmidrule(lr){7-8}
Driver & bW & (\%) & (s) & up & KB & (\%) & (s) \\
\midrule
Exact (float16) & 1.0000 & 0.67 & 42.9 & 1 & 21.0 & 0.51 & 34.35 \\
Lipschitz-pruned & 1.0000 & 0.67 & 41.6 & 0.5--3.1 & 21.0 & 0.51 & 31.65 \\
State only & 1.0041 & 3.22 & 3.6 & 1.2--163 & 26.7 & 4.63 & 0.48 \\
Two-stage $k{=}10$ & 1.0004 & 1.75 & 4.3 & 1.3--137 & 53.0 & 2.38 & 0.69 \\
Two-stage $k{=}100$ & 0.9995 & 1.05 & 4.8 & 1.3--59 & 53.0 & 1.15 & 1.04 \\
Sample $1\times$ & 1.0012 & 1.29 & 12.1 & 0.8--65 & 42.0 & 3.79 & 3.16 \\
Sample $2\times$ & 1.0000 & 0.87 & 15.0 & 0.8--56 & 42.0 & 1.57 & 6.20 \\
Sample $4\times$ & 0.9999 & 0.72 & 19.2 & 0.8--39 & 42.0 & 0.81 & 10.28 \\
Subsample 10\% & 1.0064 & 3.18 & 3.1 & 7.5--73 & 23.1 & 3.76 & 2.92 \\
\midrule
Candidates & State & $k{=}10$ & $k{=}100$ & Exact & \multicolumn{3}{r}{DuckDB speed-up} \\
\midrule
Random, $10^2$ & 6.66 & 3.71 & -- & 3.26 & \multicolumn{3}{r}{3.3--75} \\
Random, $10^3$ & 4.28 & 2.41 & 1.63 & 1.49 & \multicolumn{3}{r}{11--626} \\
Random, $10^4$ & 3.22 & 1.75 & 1.05 & 0.67 & \multicolumn{3}{r}{13--2,301} \\
Subsets, $2^K{-}1$ & 3.30 & 3.02 & -- & 2.32 & \multicolumn{3}{r}{1.4--47} \\
Local search & 3.24 & 2.26 & -- & 0.41 & \multicolumn{3}{r}{1.1--17} \\
\bottomrule
\end{tabularx}
\end{table}

\emph{Protocol.} For each panel and $K$ we form families of groups (M4: all series, category, category $\times$
forecast origin, category $\times$ chunks of 20 series; EPF: all, market, market $\times$ quarter; DE: all, quarter,
two months; the DE quarters are the groups of the running example), 604 groups in 34 families. Within every cell, the blocks outside the realistic-weights validation split
are split 75/25 into calibration and test (5 splits); groups keep at least 19 calibration and 5 test blocks. Each
group searches for the smallest calibrated width $f(w)$ (Section~\ref{sec:method}) with the same strategies for every
driver: $10^2$, $10^3$ or $10^4$ Dirichlet(1) candidates, all $2^K-1$ equal-weight subsets, or a greedy coordinate
search from equal weights. Every pick is evaluated by exact float64 conformal. Its \emph{in-sample} regret is its exact
calibrated width over the best exact width that any exact driver found, minus one. \emph{Out of sample}, each pick is
deployed with its exact shift on the held-out blocks, and we report its block-Winkler score (normalised width plus
$2/\alpha$ times the block score's excess over the shift) over that of the exact-search pick. Table~\ref{tab:v2search} reports this ratio as a mean over groups and the regret
as a pooled mean, times summed over families (builds excluded), speed-ups as float16 exact time over the driver's
(range over families) and the median bytes read per group; its bottom part gives the in-sample regret per candidate
set and the exact-SQL time over the state's SQL screen. Two
competitors use the same float16 re-check of the top $k$ as the two-stage plan but screen otherwise: a bottom-$s$ sample
of float16 score-sufficient blocks at 1, 2 or 4 times the state's bytes, or a uniform 10\% of the calibration blocks.
A third is exact search with Lipschitz pruning: the net points are evaluated exactly as anchors $w_a$, and a candidate
is skipped when its lower bound $\ip{w}{\bar\beta}+2\max_a(\qhat(w_a)-L\norm{w-w_a})$, with the constant $L\le\frac12$ of
Lemma~\ref{lem:lip} computed from the rows, reaches the best exact width so far; it returns the exact pick in every
group.

\emph{Two-stage plan against samples and pruning.} At $10^4$ candidates (Table~\ref{tab:v2search}, top) re-checking
the top 100 matches exact search out of sample: the held-out block-Winkler score of its pick is 99.95\% of that of
the exact-search pick, against 100.41\% for the state alone. In sample its regret is 1.05\%, against 0.67\% for exact
search and 3.22\% for the state. It takes 4.8~s over all families against 42.9~s, and it is faster than float16 exact
search in all 34 families, by 1.3--59 times. No screen with the same re-check is both faster and better. The sample
at the state's bytes reaches 1.29\% in 12.1~s, and in the 114 large groups, whose blocks exceed the $1\times$ sample, 3.79\% in
3.16~s against 1.15\% in 1.04~s (elsewhere it holds every block and is exact search from rows). Larger samples buy
regret with time: at 2 and 4 times the bytes they reach 0.87\% and 0.72\% at three to four times the plan's time, and
in the large groups only the $4\times$ sample is better (0.81\% in 10.28~s). The 10\% subsample is faster (3.1~s) but
loses 3.18\%. On these panels the plan is thus Pareto-optimal among these drivers, not dominant (on both real panels of $10^8$ rows it beats them in regret at less total time,
Section~\ref{sec:exp-scale}), and its gain is time, not memory: it
reads the state and the float16 rows, 54,268 bytes per group (median) against 21,516 for the rows alone. Lipschitz
pruning cannot compete on time: it evaluates 4.9--101\% of the candidates per family and takes 0.97 of the exact
time, because a candidate is skipped only when $f(w)$ minus a slack of up to $2L\norm{w-w_a}$ reaches the best width,
and in the large groups that dominate the time this slack exceeds almost every candidate's excess over the best.
Every driver of Table~\ref{tab:v2search} is 6.7--7.4\% above the test oracle (the candidate chosen with the test
labels): selecting among $10^4$ widths on the calibration blocks costs more than any difference between drivers.

\emph{Scale and SQL.} The largest group is all 755 cells of M4 Daily at $K=4$, with 140,994 calibration rows
(Figure~\ref{fig:dataflow}). There the state screens $10^4$ candidates in 42~ms and the two-stage plan finishes in
118~ms, against 6.9~s for float16 exact search; the build is paid once. On the
synthetic $\times40$ replica of M4 Daily (Section~\ref{sec:exp-cost}; cost only), searching $10^4$ candidates in its six category groups takes
the two-stage plan 6.0~s against 650~s for exact float16 search, 108 times less, but a $4\times$ sample screen with the
same re-check takes 6.8~s: at 10,652,320 rows the re-check of 100 candidates dominates both, while the state alone
screens in 49~ms. In DuckDB the state screens $10^4$ candidates 13--2,301 times faster than exact SQL
(Table~\ref{tab:v2search}, bottom; 65~ms against 148.9~s in the largest group), and its picks equal the NumPy picks in
every group (exact SQL at $10^4$ candidates ran in 30 of the 34 families).

\emph{Where exact search is enough.} For the $2^K-1$ subsets float16 recomputation is faster than the state in 5 of 34
families, and greedy local search from float16 rows is cheap and reaches 0.41\% (the state's: 3.24\%).

\emph{Coverage of searched picks.} Deployed with their exact shift, the picks of exact and two-stage searches over
$10^4$ candidates cover 0.84--0.90 of the held-out blocks (mean per family), below the target 0.9, since the minimum
over $10^4$ calibrated widths is optimistic on its own blocks. With the state's shift, the state's own picks cover
0.88--0.92 at a median width 1.01--1.10 times the exact one; Corollary~\ref{cor:sel}(i) covers this deployment when blocks are i.i.d., but with $|N|=84$ its slack is
0.022 for the largest group and 0.26--0.51 for the groups of 20--72 calibration blocks that make up most families, so
there the observed coverage comes from the net inflation, not from the corollary. A search that needs the nominal
level should recalibrate its pick on held-out blocks; the rank $k'$ of Corollary~\ref{cor:sel}(iii), with the state's
shift and i.i.d.\ blocks, is finite only for $n\ge406$ ($|N|=84$) or $n\ge427$ ($|N|=126$).

\subsection{At Scale: Smart-Meter and Retail Panels of \texorpdfstring{$10^8$}{10\textasciicircum 8} Rows}\label{sec:exp-scale}

\begin{table*}[!t]
\centering
\scriptsize
\caption{Configuration search at scale: CER and M5, $10^4$ candidates per group unless stated.}
\label{tab:v2scale}
\setlength{\tabcolsep}{3.2pt}
\begin{tabularx}{\textwidth}{@{}lRRRRRRRRRRRR@{}}
\toprule
 & \multicolumn{3}{c}{CER, $K=4$} & \multicolumn{3}{c}{CER, $K=6$} & \multicolumn{3}{c}{M5, $K=4$} & \multicolumn{3}{c}{M5, $K=6$} \\
\cmidrule(lr){2-4}\cmidrule(lr){5-7}\cmidrule(lr){8-10}\cmidrule(l){11-13}
Driver & Regret & bW & Speed-up & Regret & bW & Speed-up & Regret & bW & Speed-up & Regret & bW & Speed-up \\
\midrule
State only & 1.18 & 1.0033 & 2,007 & 1.77 & 1.0045 & 1,406 & 0.94 & 1.0004 & 1,464 & 1.58 & 1.0019 & 952 \\
Two-stage $k{=}10$ & 0.65 & 1.0012 & 415 & 1.02 & 1.0012 & 388 & 0.50 & 1.0001 & 360 & 0.79 & 0.9998 & 304 \\
Two-stage $k{=}100$ & 0.31 & 1.0001 & 54 & 0.65 & 0.9999 & 51 & 0.25 & 1.0000 & 52 & 0.50 & 0.9998 & 44 \\
Exact (float16) & 0.08 & 1.0000 & 1 & 0.49 & 1.0000 & 1 & 0.03 & 1.0000 & 1 & 0.28 & 1.0000 & 1 \\
Sample $1\times$ & 1.88 & 1.0054 & 37 & 1.94 & 1.0052 & 37 & 1.33 & 1.0013 & 36 & 1.71 & 1.0030 & 33 \\
Sample $4\times$ & 1.43 & 1.0036 & 31 & 1.57 & 1.0035 & 31 & 0.98 & 1.0006 & 29 & 1.16 & 1.0015 & 27 \\
Subsample 10\% & 1.41 & 1.0031 & 6.3 & 1.54 & 1.0031 & 6.8 & 0.98 & 1.0007 & 4.6 & 1.16 & 1.0015 & 4.4 \\
\midrule
\multicolumn{13}{@{}l}{\emph{Greedy forward selection (Caruana, $T{=}25$)}} \\
Two-stage $k{=}10$ & 0.70 & 1.0008 & 8.9 & 0.66 & 0.9996 & 11 & 0.54 & 1.0002 & 9.0 & 0.63 & 0.9995 & 11 \\
Exact (float16) & 0.14 & 1.0000 & 1 & 0.08 & 1.0000 & 1 & 0.15 & 1.0000 & 1 & 0.09 & 1.0000 & 1 \\
\bottomrule
\end{tabularx}
\end{table*}

\emph{Data and protocol.} The CER smart-metering trial~\cite{CER2012} records the half-hourly electricity use of
Irish homes and businesses. We keep the 5,670 of 6,435 meters with at most 1\% missing half-hours between 2009-07-14
and 2010-12-31; the 0.047\% missing cells are filled with the previous day's value for the members' inputs only, and
blocks whose outcome day has a filled cell are dropped. A block is one meter's forecast of the 48 half-hours of a
day from all earlier days, on each of 389 forecast days: 2,198,767 blocks and 105,540,816 rows. The members are those
of the EPF and DE pools without exogenous inputs; the global LightGBM members are fitted per chunk of about 89 meters,
refit every 91 days on at most the preceding 364 days and their scale floored at 0.01~kWh. The pilot is the first 28 forecast days; cells are (meter
type, 20-meter chunk, month), and the six families (all meters, type, month, type $\times$ month, chunk, chunk
$\times$ quarter) have 1,762 groups. M5~\cite{Makridakis2022M5} holds the daily unit sales of 30,490 Walmart
item-store series, 68\% of whose daily values are zero (Nixtla mirror; research use only). A block is one series' 28-day forecast at one of 104 weekly origins, from at
most the preceding 364 days: 3,170,960 blocks and 88,786,880 rows. The members are those of the M4 pool, with the
global LightGBM members fitted per chunk of 305 series, refit every four origins and their scale floored at one unit.
The pilot is the first four origins; cells are (store, department, 20-series chunk, quarter), and the ten families (all,
state, store, category, department, state $\times$ category, store $\times$ category, store $\times$ department, store
$\times$ department $\times$ quarter or $\times$ chunk) have 2,243 groups, with 2,286,770 calibration blocks in the
largest and a median of 1,500 in the 1,550 chunk groups. Nets, grid, drivers and the metrics of Table~\ref{tab:v2scale} are as in Section~\ref{sec:exp-search}
(regret and bW pooled over groups and splits, speed-ups as ratios of times summed over families), plus greedy forward selection with replacement~\cite{Caruana2004}
(25 steps), whose state scores of every step are re-checked for the top 10. Each family runs as one job with 16 threads on a node
of the DTU HPC cluster (Xeon E5-2650 v4, Gold 6226R or Gold 6342), the same for every driver; times are medians of
five warm repeats of the whole family, and speed-ups compare drivers within one job. The CER data are licensed by the
Irish Social Science Data Archive and are not redistributed; our code and result summaries are.

\emph{Results.} At $10^4$ candidates (Table~\ref{tab:v2scale}), re-checking the top 100 picks configurations whose
held-out block-Winkler score is within 0.02\% of the exact-search pick's on both panels and for both $K$, 44--54 times
faster in total than vectorised float16 exact search (46.4~s against 2,526~s over the six CER families and 79.1~s
against 4,081~s over the ten M5 families at $K=4$).
Unlike on the small panels, the plan beats every sample screen with the same re-check in regret in all four settings,
at less total time: at $K=4$ the sample at the state's bytes reaches 1.88\% (CER) and 1.33\% (M5) in 68.9~s and 112.3~s,
the $4\times$ sample 1.43\% and 0.98\% in 80.7~s and 138.6~s, against 0.31\% and 0.25\% for the plan. With 28--195
blocks per group, their top 100 can miss the best candidates. In the coarse families their time is close to the plan's (9.6 against 9.8~s and 8.1 against 8.2~s in
the largest groups); in the many-group families it is up to seven times longer. The 10\% subsample (1.41\% and 0.98\%) takes 401.7~s and 887.1~s. The state alone is 952--2,007 times faster than exact search, at a regret of
0.94--1.77\%. In the largest groups at $K=4$ (73,442,640 and 64,029,560 calibration rows) the states of all 3,420 and
12,400 cells are built in 11.1~s and 10.9~s; the state then screens $10^4$ candidates in 49 and 163~ms, the two-stage
plan takes 9.8 and 8.2~s and exact search 480 and 377~s. In greedy forward selection, scoring each step with the state
and re-checking its top 10 costs 0.54--0.70\% regret at 8.9--11 times less time than exact steps. Exact search
stays the most accurate in sample (0.03--0.49\% at $10^4$ candidates). Searched picks again cover slightly less than
the target: deployed with their exact shift, the picks of exact and two-stage search cover 0.896--0.897 of the
held-out blocks (pooled over groups), and the state's picks with the state's shift 0.900--0.905.

\subsection{Equal-Bytes Regime Map and Nets}\label{sec:exp-bytes}

\begin{table}[!t]
\centering
\scriptsize
\caption{Equal-bytes regime map over the 10 panel$\times K$ settings.}
\label{tab:v2regime}
\setlength{\tabcolsep}{1.0pt}
\begin{tabularx}{\columnwidth}{@{}Rccllll@{}}
\toprule
Budget & Reliable & Mean width & Sample & Net state & Best reliable & Matched PAC \\
(KB) & leader & only & blocks & ratio & other ratio & sample ratio \\
\midrule
4 & 9/1/0 & 6/4/0 & 3--23 (4) & 1.080--1.256 & 1.133--1.359 & 1.65--50.63 (8$\times\infty$) \\
8 & 9/1/0 & 4/4/2 & 7--47 (1) & 1.061--1.193 & 1.117--1.314 & 1.19--2.56 (4$\times\infty$) \\
16 & 6/4/0 & 1/7/2 & 16--95 & 1.048--1.169 & 1.041--1.243 & 1.04--1.47 (1$\times\infty$) \\
32 & 1/6/3 & 0/4/6 & 32--192 & 1.038--1.155 & 1.016--1.589 & 1.01--1.22 \\
64 & 1/3/6 & 0/1/9 & 67--397 & 1.033--1.144 & 1.000--1.131 & 1.00--1.11 \\
128--1k & 0/0/10 & 0/0/10 & 135--6356 & 1.028--1.142 & 0.994--1.010 & 0.99--1.01 \\
\bottomrule
\end{tabularx}
\end{table}

Table~\ref{tab:v2regime} compares all methods at budgets from 4 to 1,024~KB over the 10 panel $\times$ $K$ settings,
with 40 random splits, 24 post-hoc $w$ and 400 test blocks each; a method is reliable if at most 5\% of its (split,
$w$) cells fall below 0.85 coverage, and PAC versions are matched to $(\epsilon',\delta')=(0.15,0.05)$. The table counts
the settings where the net state is narrowest strictly, tied within two standard errors, or beaten; it gives the blocks
held by the largest row sample (compressed float16), with, in parentheses, the settings where no row sample reaches the
9 blocks of a finite shift, and widths as ratios to exact per-$w$ conformal. At 4 and 8~KB the net state is the narrowest reliable summary in 9 of 10
settings and tied in the other: row samples hold only a few blocks there, their PAC versions are often infinite, and
the rival that sets the bar is the net-score sample. At 16~KB the state is best or tied everywhere, at 32--64~KB it
mostly ties or loses, and from 128~KB on it loses in all 10 settings, because its net inflation floors at 1.028--1.142
while samples approach exact. The ordering holds for thresholds 0.80 and 0.85 and shares of 1, 5 and 10\%.

\begin{table}[!t]
\centering
\scriptsize
\caption{Realistic, declared and neighbourhood weights (10 panel$\times K$ settings, 50 splits).}
\label{tab:v2weights}
\setlength{\tabcolsep}{2pt}
\begin{tabularx}{\columnwidth}{@{}Lcccccc@{}}
\toprule
 & \multicolumn{2}{c}{Lattice $N_g$} & \multicolumn{2}{c}{Declared points} & \multicolumn{2}{c}{Neighbourhood} \\
\cmidrule(lr){2-3}\cmidrule(lr){4-5}\cmidrule(lr){6-7}
Weights & Ratio & Wins & Ratio & Wins & Ratio & Wins \\
\midrule
Dirichlet draws & 1.085--1.268 & 2/10 & 1.085--1.256 & 1/10 & 1.073--1.217 & 0/10 \\
Stacking (CRPS) & 1.041--1.359 & 2/10 & 1.005--1.022 & 10/10 & 1.005--1.021 & 6/10 \\
Stacking (pinball) & 1.049--1.233 & 3/10 & 1.004--1.023 & 10/10 & 1.004--1.021 & 8/10 \\
Drifting refits & 1.083--1.322 & 2/10 & 1.020--1.164 & 2/10 & 1.019--1.113 & 3/10 \\
Top-two subsets & 1.004--1.025 & 10/10 & 1.004--1.025 & 9/10 & 1.004--1.021 & 6/10 \\
Top-three subsets & 1.005--1.238 & 5/10 & 1.005--1.038 & 10/10 & 1.005--1.038 & 8/10 \\
\midrule
Bytes, $K{=}4$ & \multicolumn{2}{c}{21,844} & \multicolumn{2}{c}{22,624} & \multicolumn{2}{c}{28,864--161,984} \\
Bytes, $K{=}6$ & \multicolumn{2}{c}{32,764} & \multicolumn{2}{c}{40,304} & \multicolumn{2}{c}{40,304--120,904} \\
\bottomrule
\end{tabularx}
\end{table}

\emph{Realistic, declared and neighbourhood weights.} Table~\ref{tab:v2weights} uses weights fitted on a split
disjoint from calibration and test (30\% of the series or days, so B2 holds) and compares the state with the float16
sample at its bytes over 50 random splits. Its ratios are the state's width over exact per-$w$ conformal, averaged over
splits and the weight vectors of a type (range over settings); wins count the settings where this mean is below the
sample's. The three nets are $N_g$; $N_g$ with all equal-weight subsets and both stacking weights as declared points;
and that net plus a finer lattice in an $\ell_1$ ball around the stacking weights, fixed on the validation blocks. Off the net, stacking weights and drifting refits behave like Dirichlet
draws, and the sample is narrower in 7 or more of the 10 settings. Weights on the net favour the state: the top-two
subsets, which are points of $N_g$, in 10 of 10 settings, and the stacking weights in 10 of 10 once they are declared
as net points, where only the grid gap remains. A neighbourhood net narrows the drifting refits, but the equal-bytes
sample grows with it and stays ahead in 7 of 10 settings; it also pulls
the state's local search towards itself (regret 7.93\% against 3.24\%). No net shift fell below the exact one in any of
these runs.

\subsection{Cost of Build, Merge and Query}\label{sec:exp-cost}

\begin{figure}[!t]
\centering
\includegraphics[width=\columnwidth]{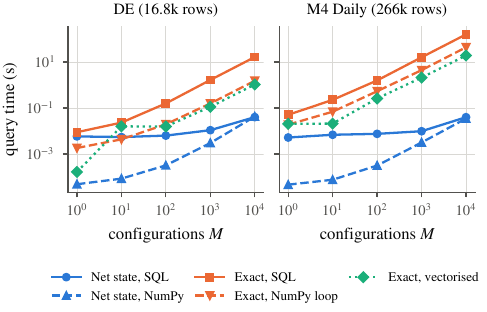}
\caption{Query time against the number $M$ of post-hoc configurations on the smallest (DE) and largest
(M4 Daily) panel, $K=4$, one state per panel: net state in SQL and NumPy, exact recomputation in SQL, NumPy loop and
vectorised NumPy.}
\label{fig:latency}
\end{figure}

\emph{One state per panel.} The build statement of Figure~\ref{fig:sql} produces the state at 93,655--170,943 rows
per second on the five panels and equals the NumPy state count for count; roll-up, complement and a 5\% append each
equal the state built from scratch. Figure~\ref{fig:latency} shows where the state pays off. The state's SQL query time
stays within a factor of two from one to $10^3$ configurations, while every exact method grows about linearly in $M$
from $10^2$ on. At $M=10^4$ the SQL query costs
3.89--6.82~$\mu$s per configuration, 413--4,039 times less than exact recomputation in SQL, and the NumPy query is
17.9--437.8 times faster than vectorised recomputation from score-sufficient rows; at $M=1$ engine overhead
dominates.

\begin{table}[!t]
\centering
\scriptsize
\caption{Many groups and panes ($K{=}4$ unless stated; medians of 5 runs).}
\label{tab:v2manygroup}
\setlength{\tabcolsep}{1.5pt}
\begin{tabularx}{\columnwidth}{@{}Lrrrrrrrr@{}}
\toprule
Panel & Cells & Sparse & Hybr. & Break- & f16, & f16, & SQL, & Coarse \\
 & & /f16 & /f16 & even & $M{=}1$ & $M{=}10^3$ & $M{=}10^3$ & ms \\
\midrule
M4 Hourly & 1,863 & 0.44 & 0.44 & 13/28 & 2.3--5.5 & 149--850 & 179--358 & 0.4/1.2 \\
M4 Daily & 19,022 & 1.48 & 1.00 & 43/98 & 0.80--1.5 & 136--384 & 238--453 & 1.3/2.5 \\
M4 Daily$^6$ & 19,022 & 1.49 & 1.00 & 45/98 & 0.53--1.0 & 95--187 & 185--386 & 1.8/3.6 \\
M4 Weekly & 1,616 & 1.59 & 1.00 & 47/105 & 0.61--1.6 & 27--115 & 60--105 & 0.5/1.8 \\
EPF & 125 & 0.41 & 0.41 & 25/57 & 0.56--6.2 & 5.7--985 & 3.3--471 & 1.5/5.6 \\
DE & 25 & 0.38 & 0.38 & 25/57 & 0.62--5.6 & 6.2--170 & 3.2--123 & 0.6/2.2 \\
M4 D.$\times$40 & 760,880 & 1.48 & 1.00 & 43/98 & 0.36--0.82 & 92--230 & 609 & 1.2/2.6 \\
\bottomrule
\end{tabularx}
\end{table}

\emph{Many groups and panes.} Table~\ref{tab:v2manygroup} keeps one state per (group, pane) cell and queries unions of
cells (roll-ups, series subsets, sliding windows of panes, all cells) with $M\in\{1,10,100,1000\}$ Dirichlet
configurations; cells hold one block on M4 and 30 on EPF and DE, M4 Daily$^6$ uses $K=6$, and M4 D.$\times$40 is
the replica described below (SQL timed once there). Sparse and hybrid are the bytes of the sparse state and of a
per-cell hybrid over float16 rows; break-even is the cell size from which the dense state is smaller than float32
or float16 rows; speed-ups are exact over state query time, NumPy against float16 rows and SQL against raw rows (range
over union families); coarse is the NumPy query time in ms of the second-level states described next at $M=1$ and $10^3$ (slowest family). The
advantage is query time for many configurations. At $M=1$ the state is slower than float16 rows
on most M4 Daily and Weekly unions, because a roll-up over one-block cells reads as many states as blocks; a second
cube level of (category or market $\times$ pane) states removes this (coarse column). The dense state is smaller than
the rows only from tens of blocks per cell. Baseline (e) answers none of the post-hoc queries. Every state, merge and
query equals its counterpart in the other engine.

\emph{Scale.} To pass $10^7$ rows we replicate the non-pilot M4 Daily blocks 40 times with perturbed member intervals
and outcomes and new series keys; the replicas are synthetic and serve for cost only. The table has 10,652,320 rows
in 760,880 cells and is built at 843,567 rows per second in NumPy and 156,996 in SQL. At $M=1000$ exact recomputation
from float16 rows takes 96,280~ms for the category roll-up, while coarse states take 2.6~ms in 655,320 bytes. On real
data, the 97,920,480 non-pilot CER rows in the 3,420 cells of Section~\ref{sec:exp-scale} ($K=4$, 64 cores of a shared
AMD EPYC 9554 node) are
built into states in 3.97~s by NumPy and 14.65~s by DuckDB with 64 threads (89.76~s and 945.65~s with one), equal
in both engines. Over all rows, $10^4$ configurations take 0.07~s from the NumPy state and 0.73~s in
SQL, against 163~s for vectorised float16 recomputation and 3,428~s in SQL. Sliding a six-month window (48,965,184
rows) by one month costs 6~ms by merge and retraction against 2.95~s to rebuild it in NumPy, and 0.45~s against
7.61~s in SQL.

\subsection{Forward in Time}\label{sec:exp-time}

\begin{table}[!t]
\centering
\scriptsize
\caption{Forward in time on EPF and DE (the 406 test steps after decay selection).}
\label{tab:v2temporalfwd}
\setlength{\tabcolsep}{1.3pt}
\begin{tabularx}{\columnwidth}{@{}Lcccccccccccc@{}}
\toprule
 & & \multicolumn{3}{c}{Net state, $\lambda{=}1$} & \multicolumn{2}{c}{Exact, $\lambda{=}1$} & \multicolumn{4}{c}{Net state, $\lambda{=}0.95$} & Raw VI \\
\cmidrule(lr){3-5}\cmidrule(lr){6-7}\cmidrule(lr){8-11}\cmidrule(lr){12-12}
Panel & $K$ & Cov. & win. & Wink. & Cov. & win. & Cov. & win. & Wink. & ESS/$\Omega$ & blk/row \\
\midrule
EPF & 4 & 0.922 & 0.026 & 2.73 & 0.906 & 0.077 & 0.924 & 0.010 & 2.78 & 195/95 & 0.334/0.813 \\
EPF & 6 & 0.931 & 0.021 & 2.63 & 0.905 & 0.078 & 0.936 & 0.007 & 2.70 & 195/95 & 0.333/0.819 \\
DE & 4 & 0.908 & 0.162 & 2.38 & 0.885 & 0.227 & 0.948 & 0.059 & 3.18 & 39/19 & 0.296/0.818 \\
DE & 6 & 0.934 & 0.103 & 2.20 & 0.898 & 0.204 & 0.953 & 0.038 & 2.75 & 39/19 & 0.300/0.823 \\
\bottomrule
\end{tabularx}
\end{table}

Table~\ref{tab:v2temporalfwd} evaluates EPF and DE strictly forward in time: each delivery day $d\ge119$ is calibrated
on days $[28,d)$, with the state updated by the decay recursion or by adding and subtracting panes, over 33--35
post-hoc $w$. The decay selected on the first 203 steps is $\lambda=1$ in all four settings. On the remaining 406 steps the net state covers 0.908--0.934 of the test blocks and
exact per-$w$ conformal 0.885--0.906; no theorem certifies these numbers (Theorem~\ref{thm:cov} is a robustness bound).
The table reports block coverage (Cov.), the share of (week, $w$) cells below 0.8 (win.), the per-row Winkler score
in units of $\sigma_o$ (Wink.), the mean effective sample size $(\sum\omega_i)^2/\sum\omega_i^2$ and total weight
$\Omega=\sum\omega_i$, and the block and row coverage of the uncalibrated $(0.1,0.9)$ interval (Raw VI). Decay with $\lambda=0.95$ lowers the share of (week, $w$) cells below 0.8 coverage from 0.021--0.162 to 0.007--0.059,
but at an effective sample size of 39--195 blocks. In 1,333,344 checks no decayed or unweighted net shift fell below the exact weighted shift; the decayed state
runs in floating point, so this is an empirical check for it, while the integer pane state is exact.

\subsection{Twelve Members and Deletion}\label{sec:exp-k12}

\begin{table}[!t]
\centering
\scriptsize
\caption{Twelve members: net state on restricted families (top) and size bounds at $\epsilon=0.1$ (bottom).}
\label{tab:v2k12}
\setlength{\tabcolsep}{2.6pt}
\begin{tabularx}{\columnwidth}{@{}Lrrrlll@{}}
\toprule
Family & $g$ & $|N|$ & KB & Coverage & Width ratio & Sample ratio \\
\midrule
pairs & 4 & 210 & 53.3 & 0.921--0.932 & 1.071--1.102 & 1.023--1.161 \\
pairs & 12 & 738 & 187.4 & 0.909--0.915 & 1.027--1.045 & 1.011--1.030 \\
triples & 4 & 870 & 220.9 & 0.931--0.949 & 1.120--1.197 & 1.005--1.093 \\
$\Delta_{12}$ & 2 & 78 & 19.8 & 0.972--0.993 & 1.653--2.364 & 1.432--1.497 (3$\times\infty$) \\
$\Delta_{12}$ & 4 & 1,365 & 346.6 & 0.962--0.987 & 1.385--1.774 & 0.990--1.061 \\
\midrule
 & \multicolumn{2}{r}{General LB} & \multicolumn{2}{r}{Net-state LB} & \multicolumn{2}{l}{Net-state UB} \\
 pairs & \multicolumn{2}{r}{2,243} & \multicolumn{2}{r}{3,565} & \multicolumn{2}{l}{46,495} \\
 triples & \multicolumn{2}{r}{5,631} & \multicolumn{2}{r}{15,841} & \multicolumn{2}{l}{2,126,461} \\
 $\Delta_{12}$ & \multicolumn{2}{r}{5,631} & \multicolumn{2}{r}{878,906,251} & \multicolumn{2}{l}{$1.31\cdot 10^{15}$} \\
\bottomrule
\end{tabularx}
\end{table}

Table~\ref{tab:v2k12} runs the state on the twelve-member pools with nets on unions of faces: a family holds all
faces with at most $p$ members ($\Delta_{12}$: full support), and ratios are medians over $w$ of the width over exact
per-$w$ conformal, ranged over the five panels, on random block splits ($20\times24$ $w$); the sample ratio is that of
the block sample at equal bytes. Its size bounds are for
tightness $\epsilon=0.1$ on the observed shift window $[0.08,1.75]$, the upper one $\min_g|N_g\cap W_p|(J+1)+1$
with $h=\epsilon-r_p(g)$ in Theorem~\ref{thm:tq}(d). On pairs the width ratio
falls from 1.071--1.102 ($g=4$, 53.3~KB) to 1.027--1.045 ($g=12$); on full support it stays at 1.385--1.774 even at
346.6~KB, and worst-case tightness 0.1 needs 878,906,251 coordinates in this class (Theorem~\ref{thm:tq}(c)). At
equal bytes the block sample is narrower in 22 of the 25 restricted panel $\times$ family cells.

\begin{figure}[!t]
\centering
\includegraphics[width=\columnwidth]{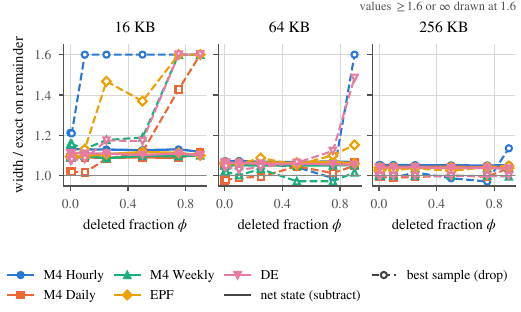}
\caption{Deleting a fraction $\phi$ of entities from a built summary at 16, 64 and 256~KB: net state after
subtraction (solid) and best pre-built block sample after dropping deleted blocks (dashed); width over exact on the
remainder, $K=4$.}
\label{fig:deletion}
\end{figure}

Figure~\ref{fig:deletion} deletes a random fraction $\phi$ of entities after the summary is built (10 splits).
Subtraction gives exactly the state of the remainder, and its width ratio moves by at most 0.027 over $\phi$ ($K=4$).
A pre-built sample can only drop deleted blocks, since refilling needs a rescan. At 16~KB the best sample is wider
than the state at every $\phi$ on three settings and from $\phi=0.25$ or $0.5$ on DE and M4 Daily; at 64 and 256~KB
mostly only for $\phi\ge0.5$. The $J=64$ pilot grid costs little (1.088--1.114 against 1.074--1.106 for a fine
grid in 672,340 bytes), and a relative-error sketch~\cite{Masson2019} per net point at its bytes is narrower in only 2 of
10 settings.

\begin{table}[!t]
\centering
\scriptsize
\caption{When to use the \calqa state, and when not (evidence in Sections~\ref{sec:exp-search}--\ref{sec:exp-k12}).}
\label{tab:deploy}
\setlength{\tabcolsep}{2pt}
\begin{tabularx}{\columnwidth}{@{}LLl@{}}
\toprule
Situation & Use & Sec. \\
\midrule
$\ge10^3$ candidates, many or large groups & state screen, re-check top $k$ & \ref{sec:exp-search}--\ref{sec:exp-scale} \\
Searches inside the engine (SQL) & state in SQL & \ref{sec:exp-search} \\
Subsets or greedy local search & exact from float16 rows & \ref{sec:exp-search} \\
Nominal coverage of a pick & held-out recalibration & \ref{sec:exp-search} \\
Small panels, time uncritical & $4\times$ sample, re-check & \ref{sec:exp-search} \\
Declared weights (stacking, subsets) & net points or stored counts (e) & \ref{sec:exp-bytes} \\
Off-net weights, $\ge32$~KB & block sample & \ref{sec:exp-bytes} \\
A few KB per group & state & \ref{sec:exp-bytes} \\
One candidate over one-block cells & float16 rows or coarse states & \ref{sec:exp-cost} \\
Deleted entities, sliding windows & state (subtract panes) & \ref{sec:exp-cost}, \ref{sec:exp-k12} \\
Full support at $K=12$ & sample, or restrict to pairs & \ref{sec:exp-k12} \\
\bottomrule
\end{tabularx}
\vspace{-5pt}
\end{table}

\section{Conclusion}\label{sec:concl}

\calqa turns conformal calibration of quantile-averaged ensembles into a mergeable, subtractable aggregate that scores
any post-hoc weight vector without reading a row, never below the exact shift and within a known excess of it. As a
screen before exact recomputation it matches exact search on held-out blocks at a fraction of the time, and on two
real panels of about $10^8$ rows it beats block-sample screens with the same re-check in regret at less total time; it saves time,
not memory (Table~\ref{tab:deploy}). Open are adaptive recalibration under drift, the exponent gap
between the general and net-histogram lower bounds, and randomized summaries.

\noindent\textbf{Availability.} \ifextended\else Full proofs: \extURL{}; \fi
code, SQL and results: \codeURL{}.

\ifextended
\appendices
\section{Proofs}\label{app:proofs}
\renewcommand{\thesubsection}{\thesection.\arabic{subsection}}
\renewcommand{\thesubsectiondis}{\thesection.\arabic{subsection}}

This appendix gives the complete proof of every result stated in Sections~\ref{sec:prelim}--\ref{sec:bounds}, in the order of the statements. Sections~\ref{app:merge} and~\ref{app:tq} use two order-statistic facts: (O1) $x_{(k)}\le s$ if and only if $\#\{i:x_i\le s\}\ge k$
(proof of Lemma~\ref{lem:lip}(iii)); (O2) if exactly $k$ values are $\le s$, then $x_{(k)}$ is the largest of them,
since they occupy the first $k$ sorted positions.

\subsection{Proof of Lemma~\ref{lem:lip} (Lipschitz transfer)}\label{app:lip}
\begin{IEEEproof}
(i) For reals, $\max(a,b)-\max(a',b')\le\max(a-a',b-b')$: if $\max(a,b)=a$ then $a-\max(a',b')\le a-a'$. Put
$\delta=w-w'$, $a=\ip{w}{\qa-y\one}$, $b=\ip{w}{y\one-\qb}$ and $a',b'$ at $w'$. Since $\sum_k\delta_k=0$ the
outcome cancels: $a-a'=\ip{\delta}{\qa}$ and $b-b'=-\ip{\delta}{\qb}$. For $x\in\R^K$ and
$c=(\max_k x_k+\min_k x_k)/2$ we have $\ip{\delta}{x}=\ip{\delta}{x-c\one}$ and $|x_k-c|\le(\max x-\min x)/2$,
so $|a-a'|,|b-b'|\le\norm{\delta}\rho_o/2$. Hence $R_o(w)-R_o(w')\le\norm{\delta}\rho_o/(2\sigma_o)\le\norm{\delta}/2$
by B3, and exchanging $w,w'$ gives (i). (ii) For $x\in\R^\ell$,
$x_{(j)}=\min_{|S|=j}\max_{i\in S}x_i$: the indices of the $j$ smallest entries attain $x_{(j)}$, and every
$j$-set contains an index with $x_i\ge x_{(j)}$ because at most $j-1$ entries are below $x_{(j)}$. So if
$\max_i|x_i-x'_i|\le\xi$, then $\max_{i\in S}x'_i\le\max_{i\in S}x_i+\xi$ for every $S$; minimising over $|S|=j$
gives $x'_{(j)}\le x_{(j)}+\xi$, and by symmetry $|x_{(j)}-x'_{(j)}|\le\xi$. Apply this to $x_i=R_{o_i}(w)$,
$x'_i=R_{o_i}(w')$ with (i). (iii) If $x_{(j)}\le s$, the $j$ indices of the smallest entries have $x_i\le s$; if
some $j$-set $S$ has $x_i\le s$ on $S$, the min-max form gives $x_{(j)}\le\max_S x_i\le s$. Hence $x_{(j)}\le s$ iff
$\#\{i:x_i\le s\}\ge j$. As $\sigma_o>0$, $R_o(w)\le s$ iff $\ell_o(w)-y\le s\sigma_o$ and $y-h_o(w)\le s\sigma_o$,
that is $y\in I_o(w,s)$ (for $s=\pm\infty$ by the conventions).
\end{IEEEproof}

\subsection{Proof of Lemma~\ref{lem:wc} (weighted conservative quantisation)}\label{app:wc}
\begin{IEEEproof}
(i) $C$ is nondecreasing, right-continuous and constant between consecutive distinct $v_i$, with $C=0$ below
$\min_i v_i$ and $C=\Omega$ from $\max_i v_i$ on. So $E=\{C\ge\need\}$ is an up-set that contains $\max_i v_i$ and
no $s<\min_i v_i$; its infimum is a jump point $v_i$, which lies in $E$ by right-continuity.
(ii) If $\need\le\Omega$, by (i) $\{s\in G:C(s)\ge\need\}=\{s\in G:s\ge r^*\}$, whose minimum is
$r^*+\gamma_G(r^*)$ if $r^*\le s_J$ and which is empty otherwise. If $\need>\Omega$, both sets are empty.
\end{IEEEproof}

\subsection{Proof of Theorem~\ref{thm:state} (state)}\label{app:state}
\begin{IEEEproof}
(i) \emph{Lower bound.} If $\shat(w)=\infty$ there is nothing to show. Otherwise take $m$ attaining the
minimum; then $t_m<\infty$ and $C_{w_m}(t_m)\ge\need$. Every block with $S_{U_i}(w_m)\le t_m$ has
$S_{U_i}(w)\le t_m+d_m(w)/2=\shat(w)$ by Lemma~\ref{lem:lip}(ii); as $\omega_i\ge0$,
$C_w(\shat(w))\ge C_{w_m}(t_m)\ge\need$, so $\need\le\Omega$ and Lemma~\ref{lem:wc}(i) gives
$\qhat(w)\le\shat(w)$. \emph{Upper bound.} If $\qhat(w)=\infty$ it is trivial. Otherwise
$C_w(\qhat(w))\ge\need$ by Lemma~\ref{lem:wc}(i), every block with $S_{U_i}(w)\le\qhat(w)$ has
$S_{U_i}(w_m)\le\qhat(w)+d_m(w)/2$, so $C_{w_m}(\qhat(w)+d_m(w)/2)\ge\need$ and
$r^*_m\le\qhat(w)+d_m(w)/2$. By Lemma~\ref{lem:wc}(ii), $t_m=r^*_m+\gamma_G(r^*_m)$, hence
$\shat(w)\le t_m+d_m(w)/2\le\qhat(w)+d_m(w)+\gamma_G(r^*_m)$ for every $m$.
(ii) $\Phi_\tau$ is a sum of per-block terms, and addition is associative and commutative; for $A'\subseteq A$,
$A=A'\uplus(A\setminus A')$. Decay follows from $\lambda^{\tau+t'-\tau_U}=\lambda^{t'}\lambda^{\tau-\tau_U}$ term by term,
and the recursion from splitting $A_{\le\tau}$ into $A_{<\tau}$ and the blocks stamped $\tau$, which enter
$\Phi_{\tau+1}$ with weight $\lambda$.
(iii) The bins $\beta_0,\dots,\beta_{j-1}$ partition $(-\infty,s_j]$, so
$\sum_{b<j}\Phi_\tau[m,b]=\sum_i\lambda^{\tau-\tau_i}\one\{S_{U_i}(w_m)\le s_j\}=C_{w_m}(s_j)$, and the last entry is
$\sum_i\lambda^{\tau-\tau_i}=\Omega$; as $t_m$ only evaluates $C_{w_m}$ on $G$, it is a function of $\Phi_\tau$. Also
$\tau_i<\tau$ and $\lambda\le1$ give $0<\omega_i\le1$. (iv) Each entry of $\Phi(A)$ counts blocks of $A$, and integer
arithmetic is exact. (v) By linearity $\Phi^\pi=\sum_r\pi_r\sum_{U\in P_r}\varphi(U)=\sum_i\omega_i\varphi(U_i)$, so, as
in (iii), $\sum_{b<j}\Phi^\pi[m,b]=C_{w_m}(s_j)$ and the last entry is $\Omega$; hence~\eqref{eq:query} computes the
$t_m$ of (i) with $\omega_{n+1}=1$. Sums and products of rationals are rational and computed exactly, and a pane
enters $\Phi^\pi$ only through the term $\pi_r\Phi(P_r)$.
\end{IEEEproof}

\subsection{Proof of Theorem~\ref{thm:cov} and Corollary~\ref{cor:exch} (coverage)}\label{app:cov}
\begin{IEEEproof}[Proof of Theorem~\ref{thm:cov} and Corollary~\ref{cor:exch}]
\emph{Step 1 (swap lemma~\cite{BarberCandes2023}).} For $v\in\R^{n+1}$ let
$\qhat(v_{1:n})=\min\{s:\sum_{j\le n}\mathrm{wt}_j\one\{v_j\le s\}\ge1-\alpha\}$ ($+\infty$ if the set is empty) and
$\mathcal E=\{v:v_{n+1}>\qhat(v_{1:n})\}$. $\mathcal E$ is Borel because, by Lemma~\ref{lem:wc}(i),
$\qhat(v_{1:n})=\min_{j\le n}\{v_j:\sum_{l\le n}\mathrm{wt}_l\one\{v_l\le v_j\}\ge1-\alpha\}$, a minimum of finitely many
Borel functions. We claim $\sum_{i=1}^{n+1}\mathrm{wt}_i\one\{v^{(i)}\in\mathcal E\}\le\alpha$. Let
$\bar F(s)=\sum_{j\le n+1}\mathrm{wt}_j\one\{v_j\le s\}$ and $q^\circ=\min\{s:\bar F(s)\ge1-\alpha\}$, finite by
Lemma~\ref{lem:wc}(i) since $\sum_j\mathrm{wt}_j=1$. For $i\le n$ the calibration part of $v^{(i)}$ has the
cumulative function $\bar F(s)-\mathrm{wt}_i\one\{v_i\le s\}-(\mathrm{wt}_{n+1}-\mathrm{wt}_i)\one\{v_{n+1}\le s\}\le\bar F(s)$ by
B1, and for $i=n+1$ it is $\bar F(s)-\mathrm{wt}_{n+1}\one\{v_{n+1}\le s\}\le\bar F(s)$. Hence
$\qhat(v^{(i)}_{1:n})\ge q^\circ$. The test coordinate of $v^{(i)}$ is $v_i$, so $v^{(i)}\in\mathcal E$ implies
$v_i>q^\circ$, and the sum is at most $\sum_i\mathrm{wt}_i\one\{v_i>q^\circ\}=1-\bar F(q^\circ)\le\alpha$.

\emph{Step 2 (domination).} Dividing by $\Omega+\omega_{n+1}$, $C_w(s)\ge\need$ if and only if
$\sum_{i\le n}\mathrm{wt}_i\one\{V^w_i\le s\}\ge1-\alpha$, so $\qhat(w)=\qhat(V^w_{1:n})$, and Theorem~\ref{thm:state}(i)
gives $\shat(w)\ge\qhat(w)$ pointwise. By Lemma~\ref{lem:lip}(iii) the coverage event is
$\{V^w_{n+1}\le\shat(w)\}$, so miscoverage implies $V^w\in\mathcal E$.

\emph{Step 3 (swap bound).} $V^w=f_w(Z)$, where $f_w$ applies $U\mapsto S_U(w)$ to each block; this map is
continuous on each piece $(\R^{2K+1})^\ell$ of the block space, hence Borel, and $f_w(Z^{(i)})=(V^w)^{(i)}$. Using
$\sum_i\mathrm{wt}_i=1$, the definition of $\dtv$ and Step~1 pointwise,
\begin{align*}
\P(V^w\in\mathcal E)&=\textstyle\sum_{i\le n+1}\mathrm{wt}_i\P(V^w\in\mathcal E)\\
&\le\textstyle\sum_{i\le n+1}\mathrm{wt}_i\P\bigl((V^w)^{(i)}\in\mathcal E\bigr)+G_w\\
&=\E\bigl[\textstyle\sum_{i\le n+1}\mathrm{wt}_i\one\{(V^w)^{(i)}\in\mathcal E\}\bigr]+G_w\\
&\le\alpha+G_w .
\end{align*}
If $w$ is random and independent of $Z$, apply this conditionally on $w$ and integrate, which gives $\E[G_w]$.

\emph{Step 4 (chain).} For Borel $f$, $\dtv(f(X),f(Y))=\sup_E|\P(X\in f^{-1}E)-\P(Y\in f^{-1}E)|\le\dtv(X,Y)$.
$\Sc$ is continuous on each piece because $\sigma_o$ is continuous and $\ge\sigma_0>0$, and $V^w_i=h_w(\Sc(U_i))$
with $h_w(((a_k,b_k))_{k\le\ell})$ the $j(\ell)$-th smallest of $\max(\ip{w}{a_k},\ip{w}{b_k})$, again continuous on each
piece. Maps applied block by block commute with the swap, so
$\dtv(V^w,(V^w)^{(i)})\le\dtv(\Sc(Z),\Sc(Z)^{(i)})\le\dtv(Z,Z^{(i)})$; multiply by $\mathrm{wt}_i$ and sum. For random $w$
the bound $G_w\le G_{\Sc}$ holds for every value of $w$, hence for $\E[G_w]$.

\emph{Step 5 (invariance).} $G_{\Sc}$ is defined without $w$, $N$, $G$ or $j$. For
$o'=(c+b\qa,c+b\qb,c+by)$ we have $\rho_{o'}=b\rho_o$ and $\mu_{o'}=b\mu_o$, so if $\sigma_0$ binds in neither row,
$\sigma_{o'}=b\sigma_o$ and $(a_{o'},b_{o'})=(a_o,b_o)$; thus $\Sc(Z)$ is unchanged as a random element. For the
corollary, exchangeability gives $Z^{(i)}\overset{d}{=}Z$, so every gap vanishes.
\end{IEEEproof}

\subsection{Proof of Corollary~\ref{cor:sel} (data-selected \texorpdfstring{$w$}{w})}\label{app:sel}
\begin{IEEEproof}
Let $F_m(s)=\P(S_U(w_m)\le s)$ for a fresh block $U$, $F_m(+\infty)=1$, $\hat F_m(s)=C_{w_m}(s)/n$ and
$Z_N=\max_m\sup_s(\hat F_m(s)-F_m(s))_+$. For any $w$ (possibly $\hat w$), let $\hat m$ be the smallest index
attaining $\shat(w)=\min_m(t_m+\frac12d_m(w))$; each $t_m$ is a measurable function of $U_{1:n}$, and $\hat m$ and
$t_{\hat m}$ at $w=\hat w$ are measurable functions of $(U_{1:n},\vartheta)$. If
$t_{\hat m}=+\infty$, coverage is sure. Otherwise $S_{U_{n+1}}(w_{\hat m})\le t_{\hat m}$ implies
$S_{U_{n+1}}(w)\le\shat(w)$ by Lemma~\ref{lem:lip}(ii), which is coverage by Lemma~\ref{lem:lip}(iii). As
$U_{n+1}$ is independent of $(U_{1:n},\vartheta)$ (the blocks are i.i.d.\ and $\vartheta$ is independent of $Z$), the
conditional coverage, i.e.\ $\mathrm{cov}(w;U_{1:n})$, the probability of coverage at $w$ over a fresh block $U_{n+1}$
with $U_{1:n}$ (and $\vartheta$) held fixed, is at least
$F_{\hat m}(t_{\hat m})\ge\hat F_{\hat m}(t_{\hat m})-Z_N\ge k/n-Z_N$, since $C_{w_{\hat m}}(t_{\hat m})\ge k$ for the rank
$k$ in use; for $k=k(n)$, $k/n\ge(1-\alpha)(n+1)/n\ge1-\alpha$. The Dvoretzky--Kiefer--Wolfowitz inequality with
Massart's constant~\cite{Massart1990}, valid for every distribution function (for general $F$ write $X=F^-(V)$
with $V$ uniform, so that $\one\{X\le s\}=\one\{V\le F(s)\}$ and $\sup_s|\hat F-F|\le\sup_v|\hat F_V(v)-v|$; since
$\hat F_m$ and $F_m$ are right-continuous, the supremum may be taken over a countable set, so $Z_N$ is measurable), and
a union bound over the net give
$\P(Z_N>x)\le2|N|e^{-2nx^2}$ for $x>0$. (i) As $x^2-\xi_0^2\ge(x-\xi_0)^2$ for $x\ge\xi_0$,
$\E Z_N=\int_0^\infty\P(Z_N>x)\,dx\le\xi_0+\int_{\xi_0}^\infty e^{-2n(x^2-\xi_0^2)}dx\le\xi_0+\int_0^\infty e^{-2nu^2}du
=\xi_0+\sqrt{\pi/(8n)}$; take expectations. (ii) The bound $1-\alpha-Z_N$ holds for every $w$ with the same $Z_N$,
and $\P(Z_N>\xi_{\delta'})\le\delta'$. (iii) On $\{Z_N\le\xi_{\delta'}\}$, $k'/n-Z_N\ge1-\epsilon'$. A larger rank only raises
every $t_m$, so $\shat$ still dominates the exact shift of rank $k(n)$, and Corollary~\ref{cor:exch} follows from
Step~2 of the previous proof.
\end{IEEEproof}

\subsection{Proof of Lemma~\ref{lem:merge} (gadgets and merge closure)}\label{app:merge}
\begin{IEEEproof}
(a) $\qa=\sigma_0\zeta\le\qb$, $\rho=\sigma_0(\max\zeta-\min\zeta)\le\sigma_0$ and $\mu=\sigma_0$, so $\sigma=\sigma_0$ and
$R(v)=\max(\ip{v}{\zeta},-\ip{v}{\zeta}-1)=\ip{v}{\zeta}$, because $\ip{v}{\zeta}\ge\min_k\zeta_k\ge-\frac12$.
(b) Suppose $n_s>n'_s$, so $n_s\ge1$; fix $H>s+\epsilon''$. Choose $N'\ge\max(n,n_s/(1-\alpha),(n-n_s+2)/\alpha)$ with $n=|A|$;
then $k':=k(N')\ge n_s$ and $N'-k'\ge\alpha N'-(2-\alpha)\ge n-n_s$. Let $A_0$ hold $k'-n_s$ rows with $R\equiv-\frac12$ and
$(N'-k')-(n-n_s)$ rows with $R\equiv H$ (part (a)); put $A_1=A\uplus A_0$, $A_2=A'\uplus A_0$, both of size $N'$.
Mergeability gives $\Phi(A_1)=\Phi(A)\oplus\Phi(A_0)=\Phi(A_2)$ (trivially if $A_0=\emptyset$), so both queries return
the same value $S$. In $A_1$ the values $\le s$ at $w$ are the $k'-n_s$ low values and the $n_s$ values of $A$
($s\ge s_{\mathrm{lo}}>-\frac12$ and $H>s$): exactly $k'$. By (O1) and (O2), $s^*_{A_1}(w)$ is the largest of them, an
$A$-value in $[s_{\mathrm{lo}},s]\subseteq I$. In $A_2$ only $k'-n_s+n'_s<k'$ values are $\le s+\epsilon''$, so
$s^*_{A_2}(w)>s+\epsilon''$ by (O1). Validity on $A_2$ gives $S>s+\epsilon''$, and tightness on $A_1$ gives
$S\le s+\epsilon<s+\epsilon''$, a contradiction.
\end{IEEEproof}

\subsection{Proof of Theorem~\ref{thm:lb} (general lower bound)}\label{app:lb}
\begin{IEEEproof}
Put $\kappa=\lfloor(B-\frac c2\Delta_{\mathcal U})/\delta\rfloor$ and $\bar s=s_{\mathrm{lo}}+\frac c2\max_{\mathcal U}|u|^2$. The
\emph{types} are $\chi=(u,j)$, $u\in\mathcal U$, $0\le j\le\kappa$, with rows $o_\chi=o_{\zeta_\chi}$ and
$\zeta_\chi=(\bar s+\frac c2|u|^2+j\delta)\one-cu$; the \emph{tests} are $\theta=(u',j')$ with query point
$u'\in\mathcal U$ and threshold $s_\theta=\bar s-\frac c2|u'|^2+j'\delta$. There are $T_{\mathcal U}$ of each.

\emph{Step 1 (legal rows, tangent identity).} $\max\zeta_\chi-\min\zeta_\chi=c(\max_ku_k-\min_ku_k)\le1$. With
$t=\max_ku_k$, $|u|^2\ge t^2$, so $\min_k\zeta_{\chi,k}\ge\bar s+\frac c2(t^2-2t)\ge\bar s-\frac c2\ge s_{\mathrm{lo}}-\frac12\ge-\frac12$.
By Lemma~\ref{lem:merge}(a) and $\sum_ku'_k=1$,
$R_\chi(u')=\bar s+\frac c2|u|^2+j\delta-c\ip{u}{u'}=\bar s-\frac c2|u'|^2+\frac c2|u-u'|^2+j\delta$, hence
$R_\chi(u')-s_\theta=\frac c2|u-u'|^2+(j-j')\delta$: the type is the tangent plane of $-\frac c2|w|^2$ at $u$, lifted
by $j\delta$.

\emph{Step 2 (lattice parity).} For $u,u'\in N_g$, $z=g(u-u')\in\Z^K$ has $\sum_kz_k=0$, and $z_k^2\equiv z_k \pmod 2$
gives $\sum_kz_k^2=2q$ with $q\in\Z_{\ge0}$, $q=0$ iff $u=u'$. So
$R_\chi(u')-s_\theta=(q(u,u')+j-j')\delta\in\delta\Z$. Fix $\epsilon''\in(\epsilon,\delta)$; no difference lies in
$(0,\epsilon'']$, so $R_\chi(u')\le s_\theta$ if and only if $R_\chi(u')\le s_\theta+\epsilon''$.

\emph{Step 3 (window).} $s_\theta\in[s_{\mathrm{lo}},s_{\mathrm{lo}}+\frac c2\Delta_{\mathcal U}+\kappa\delta]\subseteq I$, and
$R_\chi(u')\ge\bar s-\frac c2|u'|^2\ge s_{\mathrm{lo}}$ by the identity of Step~1, since $|u-u'|^2\ge0$ and $j\ge0$.

\emph{Step 4 (unimodular pattern).} Let $\Pi[\chi,\theta]=\one\{R_\chi(u_\theta)\le s_\theta\}$. By Step~2,
$\Pi[\chi,\theta]=0$ if $j_\chi>j_\theta$ (as $q\ge0$), and for $j_\chi=j_\theta$, $\Pi[\chi,\theta]=1$ iff $u_\chi=u_\theta$.
Ordered by $j$, $\Pi$ is block upper-triangular with identity diagonal blocks, so $\det\Pi=1$.

\emph{Step 5 (a).} A multiset over the types is $x\in\Z_{\ge0}^{T_{\mathcal U}}$, with counts $(\Pi^\top x)_\theta=\#\{o\in A_x:
R_o(u_\theta)\le s_\theta\}$. For $x\ne x'$ with $|x|=|x'|=n$, invertibility gives $\theta$ with, say,
$(\Pi^\top x)_\theta>(\Pi^\top x')_\theta$, and by Step~2 the latter equals $\#\{o\in A_{x'}:R_o(u_\theta)\le s_\theta+\epsilon''\}$.
Step~3 gives the hypotheses of Lemma~\ref{lem:merge}(b) at $w=u_\theta$, $s=s_\theta$, so $\Phi(A_x)=\Phi(A_{x'})$ would
contradict it. Hence $\Phi$ is injective on the $\binom{n+T_{\mathcal U}-1}{T_{\mathcal U}-1}$ multisets of size $n$.

\emph{Step 6 (b).} Let $\Gamma\in\Z^{D\times T_{\mathcal U}}$ hold the increments of the types. If
$D+1<T_{\mathcal U}$, the integer matrix $[\Gamma;\one^\top]$ has a nonzero rational kernel vector; clearing
denominators gives a nonzero integer kernel vector $z$, and $x=z^+$ and $x'=z^-$ are distinct, nonzero, of equal size,
with $\Gamma x=\Gamma x'$, so $\Phi(A_x)=\Phi(A_{x'})$, contradicting (a). So $D\ge T_{\mathcal U}-1$.

\emph{Step 7 (c).} Take $c=1$ and $\mathcal U=N_g\cap\Delta_F$, so $|\mathcal U|=\binom{g+d}{d}$ (compositions of
$g$ into $d+1$ nonnegative parts); $g\ge1$ and $g<\epsilon^{-1/2}$, so $\delta=g^{-2}>\epsilon$. On $\Delta_F$,
$\frac1{d+1}\le|u|^2\le1$, so $\Delta_{\mathcal U}<1\le2B$ and $\kappa\ge\lfloor(B-\frac12)g^2\rfloor$. With
$\lfloor x\rfloor+1\ge(x+1)/2$, $\binom{g+d}{d}\ge g^d/d!$ and $g\ge\epsilon^{-1/2}-1\ge\frac12\epsilon^{-1/2}$ (as
$\epsilon\le\frac14$) the $\Omega_d$ form follows.
\end{IEEEproof}

\subsection{Proof of Lemma~\ref{lem:round} (support-preserving rounding)}\label{app:round}
\begin{IEEEproof}
Let $x=gw$, $f=\lfloor x\rfloor$, $\phi=x-f\in[0,1)^K$ ($\phi=0$ off the support $F$) and $r=g-\sum_kf_k=\sum_{k\in F}\phi_k$,
an integer with $0\le r\le p-1$. If $r=0$, $u=w$. Otherwise let $S$ index the $r$ largest $\phi_k$; they are
positive, since $r$ or fewer positive entries below one would sum to less than $r$. So $u=(f+\one_S)/g\in N_g$ sums to
one and $\mathrm{supp}\,u\subseteq F$ (largest-remainder rounding is studied in~\cite{ContHeidari2015}). Its error is
$g\norm{w-u}=\sum_{k\in S}(1-\phi_k)+\sum_{k\notin S}\phi_k=2(r-\sum_S\phi_k)$; the mean of the $r$ largest of
the $p$ values is at least their mean $r/p$, so $\sum_S\phi_k\ge r^2/p$ and $\norm{w-u}\le2r(p-r)/(pg)$, maximised
over integers at $r=\lfloor p/2\rfloor$. For even $p$, $r_p(g)=p/(2g)$; for odd $p$, $(p^2-1)/(2pg)$; both increase
in $p$, and $\frac{p^2-1}{2p}\le\frac{p+1}2$ and $\frac p2\le\frac{p(p+2)}{2(p+1)}$ cover the parity changes.
\end{IEEEproof}

\subsection{Proof of Theorem~\ref{thm:tq} (net histograms)}\label{app:tq}
\begin{IEEEproof}
(a) For each of the $2^K-2$ non-constant $\varsigma\in\{\pm1\}^K$ let
$\zeta_\varsigma=\frac12\varsigma+(s_0-\frac12\ip{\varsigma}{w_0})\one$: its range is $1$ and its minimum is
$\ge-\frac12+s_0-\frac12\ge-\frac12$, so by Lemma~\ref{lem:merge}(a) $R_\varsigma(v)=s_0+\frac12\ip{\varsigma}{v-w_0}$. For
$x=v-w_0$ ($\sum x=0$), $\max_\varsigma\ip{\varsigma}{x}=\norm{x}$: if $x\ne0$ it has entries of both signs and its
sign pattern is non-constant. Take $n$ with $2^K-2\le k(n)\le n$; it exists because $k(n)\ge(1-\alpha)(n+1)\to\infty$
and $n-k(n)\ge\alpha n-(2-\alpha)\ge0$ for $n\ge(2-\alpha)/\alpha$. Let $A$ hold the sign rows, $k(n)-2^K+2$ rows
$R\equiv-\frac12$ and $n-k(n)$ rows $R\equiv H>s_0+1$. At $v$ the largest sign value $x^\circ=s_0+\frac12\norm{v-w_0}\ge\frac12$
has exactly $k(n)$ values of $A$ at or below it, so $s^*_A(v)=x^\circ$ by (O2). By Lemma~\ref{lem:wc}(ii),
$t_m\ge r^*_m=s_0+\frac12 d_m(w_0)$, so $\shat(w_0)\ge\min_m(s_0+d_m(w_0))=s^*_A(w_0)+\mathrm{dist}(w_0,N)$.
(b) Take $n$ copies ($k(n)\le n$) of a row with $R\equiv x$ (Lemma~\ref{lem:merge}(a), $x\ge s_{\mathrm{lo}}>-\frac12$). Then
$r^*_m=x$, $t_m=x+\gamma_G(x)$ by Lemma~\ref{lem:wc}(ii), and $\shat(w)=x+\gamma_G(x)+\frac12\mathrm{dist}(w,N)$.
(c) Part (a) with $s_0=\max(s_{\mathrm{lo}},\frac12)\in I$ gives $\mathrm{dist}(w,N)\le\epsilon$, and (b) gives $\gamma_G(x)\le\epsilon$
on $I$. Let $s'_1<s'_2<\dots$ be the grid points in $[s_{\mathrm{lo}},s_{\mathrm{lo}}+B+\epsilon]$; $s'_1\le s_{\mathrm{lo}}+\epsilon$
since $\gamma_G(s_{\mathrm{lo}})\le\epsilon$. If $s'_i<s_{\mathrm{lo}}+B$, then for every $x\in(s'_i,s_{\mathrm{lo}}+B]$ the grid has a point in
$[x,x+\epsilon]\subseteq[s_{\mathrm{lo}},s_{\mathrm{lo}}+B+\epsilon]$, so $s'_{i+1}$ exists and $s'_{i+1}\le x+\epsilon$; letting
$x\downarrow s'_i$ gives $s'_{i+1}\le s'_i+\epsilon$, hence $s'_i\le s_{\mathrm{lo}}+i\epsilon$ by induction. The first index $L$ with
$s'_L\ge s_{\mathrm{lo}}+B$ therefore has $L\ge\lceil B/\epsilon\rceil$. For a face $\Delta_F$, $F=\{f_0,\dots,f_d\}$, the map
$\pi(w)=(w_{f_1},\dots,w_{f_d})$ sends $\Delta_F$ onto $\{x\ge0,\sum x\le1\}$ of volume $1/d!$, and for every $u\in\Delta_K$,
$\norm{w-u}\ge\sum_{i\le d}|w_{f_i}-u_{f_i}|$, so the image of $\{w\in\Delta_F:\norm{w-u}\le\epsilon\}$ lies in an
$\ell_1$ ball of volume $(2\epsilon)^d/d!$. These images cover the simplex, so $|N|(2\epsilon)^d\ge1$. For $\Delta_F$, $F\in\mathcal F$, let
$D_F=\{w\in\Delta_F:w_k\ge2\epsilon\ \forall k\in F\}=2\epsilon\one_F+(1-2(d+1)\epsilon)\Delta_F$ and let $N_F$ hold the net
points within $\epsilon$ of some point of $D_F$. If $u\in N_F\cap N_{F'}$ for faces $F\ne F'$ of equal size, pick
$k\in F\setminus F'$, $k'\in F'\setminus F$, $w\in D_F$ and $w'\in D_{F'}$ within $\epsilon$ of $u$; then
$2\epsilon\ge\norm{u-w}+\norm{u-w'}\ge|w_k-u_k|+u_k+u_{k'}+|w'_{k'}-u_{k'}|\ge w_k+w'_{k'}\ge4\epsilon$, which is
impossible. So the $N_F$ are disjoint, and the volume argument on $\pi(D_F)$, of volume $(1-2(d+1)\epsilon)^d/d!$,
gives $|N_F|\ge(1-2(d+1)\epsilon)^d(2\epsilon)^{-d}$; summing over the disjoint $N_F$ gives the bound on $|N|$. The
state has $|N|(|G|+1)+1$ coordinates.
(d) Let $w\in W$, $\mathrm{supp}\,w\subseteq F'$, $|F'|\le p$, and $u$ its rounding (Lemma~\ref{lem:round}; $u=w$ at a vertex).
Then $u\in N$ and $\norm{w-u}\le r_p(g)=r$ by monotonicity. Every block score moves by at most $r/2$ between $w$ and $u$
(Lemma~\ref{lem:lip}(ii)), hence so does the $k$-th smallest (proof of Lemma~\ref{lem:lip}(ii)), and
$r^*_u\in[s_{\mathrm{lo}}-\frac r2,s_{\mathrm{lo}}+B+\frac r2]\subseteq[s_1,s_J]$ since $(J-1)h\ge B+r$; so
$\gamma_G(r^*_u)\le h$ and Theorem~\ref{thm:state}(i) gives the sandwich. With $g=\lceil(d+1)/\epsilon\rceil$,
$r\le\frac{d+1}{2g}\le\frac\epsilon2$ and $r+h\le\epsilon$; $|N|=\binom{g+d}{d}\le(2(d+1)/\epsilon)^d/d!$ and
$J\le2B/\epsilon+3\le5B/\epsilon$, which with (c) gives $\Theta_d(B\epsilon^{-(1+d)})$.
\end{IEEEproof}

\fi

\section*{Acknowledgment}
ChatGPT was used for polishing the language of this manuscript. The authors reviewed the content and take full responsibility for it.

\bibliographystyle{IEEEtran}
\bibliography{refs/references}

\end{document}